\documentclass[]{article}
\usepackage{xcolor}
\usepackage{amssymb}
\usepackage{amsmath,amssymb,amsfonts}
\usepackage{algorithmic}
\usepackage{graphicx}
\usepackage{textcomp}
\usepackage{amsthm}
\usepackage{xspace}                           % "intelligent" spaces
\usepackage{booktabs}                                     % pretty tables
\usepackage{algorithm,algorithmic,dsfont}
\usepackage{tikz}
\usetikzlibrary{arrows,patterns,plotmarks,shapes,snakes,er,3d,automata,backgrounds,topaths,trees,petri,mindmap}
\usepackage{url}
\usepackage{subcaption}
\usepackage{balance}
\usepackage{authblk}
\newcommand{\oneplusone}{(1+1)~EA\xspace}

\newcommand{\RLS}{RLS\xspace}

\newcommand{\EBEAE}{(1+1)~EA$_{e}$\xspace}

\newcommand{\ab}{\hspace{0.125em}}                        % 1/8 em space
\newcommand{\ie}{\hbox{i.\ab e.}\xspace}                  % i.e.
\newcommand{\eg}{\hbox{e.\ab g.}\xspace}                  % i.e.
\newcommand{\etal}{et al.\ }

\newtheorem{theorem}             {Theorem}
\newtheorem{lemma}      [theorem]{Lemma}
\newtheorem{corollary}  [theorem]{Corollary}

\title{Runtime Analysis of RLS and (1+1) EA for the Dynamic Weighted Vertex Cover Problem}

\author{Mojgan Pourhassan\thanks{Corresponding author}}
\author{Vahid Roostapour}
\author{Frank Neumann}
\affil{\small Optimisation and Logistics, School of Computer Science \\The University of Adelaide, SA, Australia\\ \{mojgan.pourhassan,vahid.roostapour,frank.neumann\}@adelaide.edu.au}

\providecommand{\keywords}[1]{\textbf{\textit{Keywords---}} #1}

\begin{document}

\maketitle

\begin{abstract}

	In this paper, we perform theoretical analyses on the behaviour of an evolutionary algorithm and a randomised search algorithm for the dynamic vertex cover problem based on its dual formulation. The dynamic vertex cover problem has already been theoretically investigated to some extent and it has been shown that using its dual formulation to represent possible solutions can lead to a better approximation behaviour. We improve some of the existing results, \ie we find a linear expected re-optimization time for a \oneplusone to re\nobreakdash-discover a $2$\nobreakdash-approximation when edges are dynamically deleted from the graph. 
	Furthermore, we investigate a different setting for applying the dynamism to the problem, in which a dynamic change happens at each step with a probability $P_D$.
	We also expand these analyses to the weighted vertex cover problem, in which weights are assigned to vertices and the goal is to find a cover set with minimum total weight. Similar to the classical case, the dynamic changes that we consider on the weighted vertex cover problem are adding and removing edges to and from the graph. %For both dynamic settings, we prove that two randomised search heuristics, a randomised local search and a \oneplusone, starting with a 2\nobreakdash-approximate solution, can find a solution that is 2\nobreakdash-approximate for the new situation in expected pseudo-polynomial time. For the second setting, we also consider starting from a solution in which all the weights are set to 0.
	We aim at finding a maximal solution for the dual problem, which gives a 2-approximate solution for the vertex cover problem. This is equivalent to the maximal matching problem for the classical vertex cover problem.
\end{abstract}

\section{Introduction}
Evolutionary algorithms~\cite{eiben2003introduction}  and other bio-inspired algorithms have been widely applied to combinatorial optimization problems. They are easy to implement and have the ability to adapt to changing environments. Because of this, evolutionary algorithms have been widely applied to dynamic optimization problems~\cite{branke2012evolutionary,nguyen2012continuous}. Most studies in this area consider dynamically changing fitness functions~\cite{nguyen2012evolutionary}. However, often resources such as the number of trucks available in vehicle routing problems may change over time while the overall goal function, e.g. maximize profit or minimize cost, stays the same.

Evolutionary algorithms for solving dynamic combinatorial optimization pro-blems have previously been theoretically analysed in a number of articles~\cite{DrosteDynamic, WittDynamic,NeumannWitt2015DynamicMakespan,UsDVCGecco2015, FengLinearFunctionDynamicConstraintsGecco2017, StanhopeDynamic}. Different analyses may consider the impact of different parameters such as diversity, frequency or magnitude of the changes on the performance of evolutionary algorithms \cite{oliveto2015analysis,rohlfshagen2009dynamic}. Some of the classical problems that have been investigated in the dynamic context are the OneMax problem, the makespan scheduling problem and the vertex cover problem~\cite{DrosteDynamic, WittDynamic,NeumannWitt2015DynamicMakespan,UsDVCGecco2015}. In a recent work~\cite{FengLinearFunctionDynamicConstraintsGecco2017}, the behaviour of evolutionary algorithms on linear functions under dynamically changing constraints is investigated. 

We contribute to this area of research by investigating the (weighted) vertex cover problem in terms of its dual formulation which becomes a maximal matching problem. The vertex cover problem has the constraint that all edges have to be covered by a feasible solution. We investigate the behaviour of evolutionary algorithms when this constraint changes through the addition and removal of edges.
In~\cite{UsDVCGecco2015} the vertex cover problem is considered with a simple dynamic setting where the rate of dynamic changes is small enough, so that the studied algorithms can re-optimize the problem after a dynamic change, before the following change happens. We call this dynamic setting \emph{One-time Dynamic Setting}.
The article by Droste~\cite{DrosteDynamic} on the OneMax problem presents another setting for dynamically changing problems, where a dynamic change happens at each step with probability $p'$. We call this dynamic setting \emph{Probabilistic Dynamic Setting}. In that article, the maximum rate of dynamic changes is found such that the expected optimization time of \oneplusone remains polynomial for the studied problem. In his analyses the goal is to find a solution which has the minimum Hamming distance to an objective bit-string and one bit of the objective bit-string changes at each time step with a probability $p'$; which results in the dynamic changes of the fitness function over time. The author of that article has proved that the \oneplusone has a polynomial expected runtime if $p'=O(\log(n)/n)$, while for every substantially larger probability the runtime becomes superpolynomial. The results of that article hold even if the expected re-optimization time of the problem is larger than the expected time until the next dynamic change happens. K\"{o}tzing \etal\cite{WittDynamic} have reproved some of the results of~\cite{DrosteDynamic} using the technique of drift analysis, and have extended the work to search spaces with more than two values for each dimension. Furthermore, they analyse how closely their investigated algorithm can track the dynamically moving target over time.

In this paper, we consider both dynamic settings and analyse two simple randomised algorithms on the vertex cover problem. This paper is an extension to a conference paper~\cite{DynamicVertexCoverFOCI2017}, in which the classical vertex cover problem had been investigated. Here, we expand those analyses to the weighted vertex cover problem, where integer weights are assigned to vertices, and the goal is to find a set of vertices with minimum weight that covers all the edges. 

Different variants of the classical randomised local search algorithm (\RLS) and \oneplusone have previously been investigated for the static vertex cover problem in the context of approximations. This includes a node-based representation examined in \cite{DBLP:journals/ec/FriedrichHHNW10,KratschFrank2013VertexCover,DBLP:journals/tec/OlivetoHY09, PPSN2016WeightedVCP} as well as a different edge-based representation analysed in~\cite{DBLP:conf/foga/JansenOZ13} and a generalization of that for the weighted vertex cover problem analysed in~\cite{Us2017FogaDualVC}. 

For the dynamic version of the problem, 
three variants of those randomised search heuristics have been investigated in~\cite{UsDVCGecco2015}. The investigated variants include an approach with the classical node-based representation in addition to two approaches with edge-based representation introduced in~\cite{DBLP:conf/foga/JansenOZ13}. One of the edge-based approaches uses a standard fitness function, while the other one uses a fitness function that gives a large penalty for adjacent edges. The latter approach finds a $2$\nobreakdash-approximation from scratch in expected time $O(m\log m)$~\cite{DBLP:conf/foga/JansenOZ13}, where $m$ is the number of edges. Having the large penalty for adjacent edges in that approach results in finding a maximal matching, which induces a $2$\nobreakdash-approxi\-mate vertex cover. 
Considering the dynamic version of the problem where a solution that is a maximal matching is given before the dynamic change, Pourhassan \etal\cite{UsDVCGecco2015} have proved that the RLS re\nobreakdash-optimises the solution in expected time $O(m)$ \ie after a dynamic change, it takes expected time $O(m)$ to recompute a 2\nobreakdash-approximate solution. They also proved that \oneplusone manages to maintain the quality of $2$\nobreakdash-approximation in expected time $O(m)$ when the dynamic change is adding an edge. But for edge deletion, the expected time $O(m \log m)$ was obtained, which is the same as the expected time of finding a $2$\nobreakdash-approximate solution from scratch.

In this paper we improve the upper bound on the expected time that \oneplusone with the third approach requires to re\nobreakdash-optimise the $2$\nobreakdash-approximation when edges are dynamically deleted from the graph. We improve this bound to $O(m)$, $m$ being the number of edges, which can be shown to be tight for this problem. Moreover, we investigate the probabilistic dynamic changes for applying dynamism on the problem, in which a dynamic change happens with a certain probability, $P_D$, at each step. For the classical vertex cover problem, we prove that when $P_D$ is small enough, \oneplusone with the third approach finds a $2$\nobreakdash-approximate solution from an arbitrary initial solution in expected polynomial time, and rediscovers a solution with the same quality in expected linear time after a dynamic change happens. 

Using similar arguments, we also find pseudo-polynomial upper bounds on the expected time that \RLS and \oneplusone require to re\nobreakdash-optimise the 2\nobreakdash-approxi-mation for the dynamic weighted vertex cover problem in both dynamic settings. In the setting with probabilistic dynamic changes, we also obtain a pseudo-polynomial upper bound for the expected time that these two algorithms need to find a 2\nobreakdash-approximate solution by starting from a solution that assigns a weight of 0 to all edges. 

Similar to the classical vertex cover problem, in the weighted vertex cover problem, when we start with a 2 approximate solution, and the input graph faces a dynamic change, the solution may become infeasible, but using the edge-based representation, it is not far from a new 2 approximate solution. Similar situations may happen in other dynamic optimization problems, where the re-optimization time is usually less than the time that is required to optimize the problem from an arbitrary initial solution. Nevertheless, in the weighted version of the problem, the presence of weights has made our final results to be pseudo polynomial, rather than polynomial. 

A number of strategies have been proposed and studied for selecting the mutation strength or step size adaptation for multi-value decision variables~\cite{DoerrStepSizeRandom2016, WegenerStepSize2010}. Step size adaptation is a promising approach for finding a polynomial bound in such situations. This technique is studied in~\cite{Us2017FogaDualVC} on the static version of the weighted vertex cover  problem with a simple fitness function that aims at finding a maximal solution for the dual problem. Step size size adaptation has been proved to improve the efficiency of a randomised local search algorithm in that paper. Using this technique for the dynamic version of the problem and a fitness function that prioritises minimizing the number of uncovered edges, is left for future work on this topic. The focus of this paper is dealing with the damage that can be caused by a dynamic change, and also dealing with the situation where the distance to a maximal dual solution is increased, but the number of uncovered edges is decreased.
Moreover, since in the dual setting we are looking for a maximal solution rather than a maximum solution, the goal that the algorithm is moving towards can change when multiple mutations happen at the same step. This makes the analysis much harder for the \oneplusone, for which our resulting upper bound is presented with respect to the weight of the optimal solution of the vertex cover problem in addition to the number of edges.

The rest of the paper is structured as follows. The problem definition and the investigated algorithm are given in Section~\ref{sec:prel}. All analyses for the classical vertex cover problem are presented in Section~\ref{sec:classicalVCP}, where Section~\ref{sec:improveDelete} includes the analysis for improving the expected re-optimization time of \oneplusone with the third approach for one-time dynamic deletion of an edge, and Section~\ref{sec:newDynamicSetting} includes the investigations on the probabilistic dynamic setting for the problem. The dynamic weighted vertex cover problem is analysed in Section~\ref{sec:weightedVCP}, with the one-time and the probabilistic dynamic settings being investigated in Section~\ref{sec:WVCDyn1} and Section~\ref{sec:WVCDyn2}, respectively.
Finally, we conclude in Section~\ref{sec::conclusion}. The analyses of Section~\ref{sec:classicalVCP} are based on the conference version of this work~\cite{DynamicVertexCoverFOCI2017}.

\section{Preliminaries}
\label{sec:prel}
In this section we present the definition of the problems and the algorithms that are investigated in this paper. We divide the section into two parts. In the first part (Section~\ref{sec:prel:VCP}), we give the formal definition of the vertex cover problem and the dynamic version of that problem. Moreover, we explain the edge based approach for solving this problem and present the  algorithm that we investigate in this paper for that problem: \oneplusone. In the second part (Section~\ref{sec:prel:WVCP}), we introduce the weighted vertex cover problem, its dynamic version, the investigated approach and the  algorithms that we are analysing in this paper: \RLS and \oneplusone.

\subsection{The Dynamic Vertex Cover Problem and the Investigated Algorithms}
\label{sec:prel:VCP}
For a given graph $G=(V,E)$ with set of vertices $V=\{v_1,\dots,v_n\}$ and set of edges $E=\{e_1,\dots,e_m\}$, the vertex cover problem is to find a subset of nodes $V_C\subset V$ with minimum cardinality, that covers all edges in $E$, i.e. $\forall e \in E , e\cap V_C \ne \emptyset$.

In the dynamic version of the problem, an arbitrary edge can be added to or deleted from the graph. We investigate two different settings for applying the dynamism on the problem. In the one-time dynamic setting, which has previously been analysed in~\cite{UsDVCGecco2015}, the changes on the instance of the problem take place every $\tau = poly(n)$ iterations where $poly(n)$ is a polynomial function in $n$. We improve some results that were obtained in~\cite{UsDVCGecco2015} for this setting. In the probabilistic dynamic setting, a probabilistic dynamic change happens at each step with a probability $P_D$; therefore, in expectation, a dynamic change happens on the graph each $\frac{1}{P_D}$ steps.

For solving the vertex cover problem by means of evolutionary algorithms, two kinds of representation have been suggested: the node-based representation and the edge-based representation. While the node-based representation is the natural one for this problem, and is used in most of the relevant works~\cite{DBLP:journals/ec/FriedrichHHNW10,DBLP:journals/tec/OlivetoHY09,KratschFrank2013VertexCover}, the edge-based representation, introduced by Jansen \etal\cite{DBLP:conf/foga/JansenOZ13}, has been suggested to speed up the approximation process. In their work~\cite{DBLP:conf/foga/JansenOZ13}, they have proved that an evolutionary algorithm using the edge-based representation and a specific fitness function, can find a $2$\nobreakdash-approximate solution in expected time $O(m\log m)$ where $m$ is the number of edges in the graph.

In this representation, each solution is a bit string $s\in\{0,1\}^m$, describing a selection of edges $E(s) = \{e_i \in E \mid  s_i = 1\}$. Then the cover set of $s$, denoted by $V_C(s)$, is the set of nodes on both sides of each edge in $E(s)$. It should be noticed that the size of the solution may change according to the dynamic changes of the graph. In our analysis $m$ is the maximum number of edges in the graph.

The specific fitness function that Jansen \etal\cite{DBLP:conf/foga/JansenOZ13} have suggested for this representation is:
\begin{align}\nonumber
f(s) =& |V_C(s)| + (|V|+1)\cdot |\{e\in E\mid  e\cap V_C(s) = \emptyset\}|\\ \nonumber
&+(|V|+1)\cdot(m+1)\cdot\\
&|\{(e,e^\prime)\in E(s)\times E(s) \mid  e\ne e^\prime, e\cap e^\prime \neq \emptyset\}|\label{eq:fit_func}\text{.}
\end{align}

The goal of the studied evolutionary algorithm is to minimize $f(s)$ which consists of three parts. The first part is the size of the cover set that we want to minimize. The second part is a penalty for edges that solution $s$ does not cover, and the third part is an extra penalty inspired from the fact that a maximal matching induces a $2$\nobreakdash-approximate solution for the vertex cover problem. 

Pourhassan \etal\cite{UsDVCGecco2015} proved that an RLS with the edge-based representation and the fitness function given in Equation \ref{eq:fit_func} re\nobreakdash-discovers the $2$\nobreakdash-approximate solution if the initial solution is a maximal matching
%, which is also a 2\nobreakdash-approximation solution,
in expected time $\mathcal{O}(m)$ and such result holds for \oneplusone  if changes are limited to adding edges. For \oneplusone and dynamic deletion of an edge, the expected time $\mathcal{O}(m \log m)$ was obtained there, which is not tight. 
This bound is improved in this paper to the tight bound of $\mathcal{O}(m)$. The \oneplusone of~\cite{UsDVCGecco2015} for the edge-based representation is presented in Algorithm~\ref{alg:EdgeEA}. In the dynamic setting that was studied in that paper, a large gap of $\tau = poly(n)$ iterations was assumed in which no dynamic changes happened. %We use this setting for our analysis in Section~\ref{sec:improveDelete}. In Section~\ref{sec:newDynamicSetting}, 
In addition to analysing this dynamic setting, in this paper we consider a second setting for applying the dynamism on the problem where a dynamic change happens at each step with a certain probability.  

\begin{algorithm}[t]
	\caption{Edge-Based \oneplusone (\EBEAE)~\cite{UsDVCGecco2015}}
	\begin{algorithmic}[1]
		\STATE The initial solution, $s$, is given: a bit-string of size $m$ which used to be a $2$\nobreakdash-approximate solution before changing the graph;
		\WHILE {Stopping criteria not met}
		\STATE {Set $s':=s$;
			\STATE Flip each bit of $s'$ independently with probability $\frac{1}{m}$;
			\IF {$f(s')\leq f(s)$} \STATE {$s:=s'$;} \ENDIF
		}
		\ENDWHILE
	\end{algorithmic}
	\label{alg:EdgeEA}
\end{algorithm}

\subsection{The Dynamic Weighted Vertex Cover Problem and the Investigated Algorithms}
\label{sec:prel:WVCP}

In the weighted vertex cover problem, 
the input is a graph $G=(V,E)$ with vertex set $V=\{v_1,\ldots, v_n\}$ and edge set $E=\{e_1,\ldots, e_m\}$, in addition to a positive weight function $w\colon V \rightarrow \mathds{N}^+$ on the vertices. In this version of the problem, the goal is to find a subset of nodes, $V_C \subseteq V$, that covers all edges and has minimum weight, \ie the problem is to minimize $\sum_{v\in V_C} w(v)$, s.t. $\forall e \in E, e \cap V_C \neq \emptyset$. For the dynamic weighted vertex cover problem, similar to the classical case, we consider dynamic changes of adding and removing edges to and from the graph.

A generalization of the edge-based approach of the classical vertex cover problem for the weighted vertex cover problem has been studied in~\cite{Us2017FogaDualVC}, where the relaxed Linear Programming (LP) formulation of the problem is considered as the primal LP problem, and the dual form (which is also an LP problem) is solved by an evolutionary algorithm.
Using the standard node-based representation, in which a solution is denoted by a bit-string $x=(x_1, \cdots,x_n)$ and each node $v_i$, $i \in \{1,\cdots, n\}$ is chosen iff $x_i=1$, the Integer Linear Programming formulation of the weighted vertex cover is:
\begin{eqnarray*}
	&& \min \sum_{i=1}^n w(v_i)\cdot x_i \\
	s.t. && x_i+x_j\geq 1 \ \ \ \forall (i,j)\in E \\
	&& x_i\in \{0,1\} \ \ \ \forall i \in \{1, \cdots , n \}\text{.}
\end{eqnarray*}

By relaxing the constraint on $x$ to $x\in [0,1]$, an LP problem is obtained, and the dual form of that problem is formulated as the following, where $s_j \in \textbf{N}^+$ denotes a weight on edge $e_j$
\begin{eqnarray*}
	&& \ \ \ \ \ \max \sum_{j=1}^m s_j \\
	s.t. && \sum_{j\in \{1, \cdots, m\} \mid e_j \cap \{v\} \neq \emptyset} s_j \leq w(v) \ \ \ \ \forall v \in V\text{.}
\end{eqnarray*}
The dual problem is to maximize the weights on the edges, and the constraint is that the weight of each node must be more than or equal to the sum of weights on edges connected to that node. We say that a node is tight, when the weight of the node is equal to the sum of weights on edges connected to that node. Observe that in a maximal solution for the dual problem, at least one node of each edge is tight. Therefore, the set of tight nodes in a maximal dual solution, \ie 
$$V_C= \left\{v\in V \mid w(v)=\sum_{j\in \{1, \cdots, m\} \ \mid \ e_j \cap \{v\} \neq \emptyset} s_j\right\}$$
is a vertex cover for the primal problem, and the total weight of this solution is at most twice the sum of weights of the edges.

It is already known that when the primal problem is a minimization problem, any feasible solution of the dual problem gives a lower bound of the optimal solution of the primal problem (See~\cite{ApproximationAlgorithms} for the Weak Duality Theorem). Therefore, $\max \sum_{j=1}^m s_j$ of a maximal solution of the dual problem, is less than or equal to the weight of the optimal solution of the weighted vertex cover problem. Therefore, the vertex cover that is induced by a maximal dual solution, $V_C$, which has a weight of at most $2\cdot \max \sum_{j=1}^m s_j$, has an approximation ratio of at most 2.

In this paper, we investigate the behaviour of \oneplusone and \RLS with the edge-based approach in achieving a 2\nobreakdash-approximate solution for the weighted vertex cover by finding a maximal solution for the dual problem. The solution's representation and the mutation operator in these algorithms, which are presented in Algorithms \ref{alg:WeightedEA} and \ref{alg:WeightedRLS}, are different from what we have for the classical vertex cover problem. A solution is an integer array and represents the weights on the edges and a mutation on an edge increases or decreases this weight. The \RLS algorithm chooses a uniformly random position of the solution in each iteration, and decreases the value of it (which is the weight of the corresponding edge) with probability $\frac{1}{2}$, and increases it otherwise. The \oneplusone algorithm uses the same mutation operator but mutates each edge with probability $1/m$. Finally both algorithms accept the mutated solution $s^\prime$ if it has a strictly greater fitness value than $s$.
\begin{algorithm}[t]
	\caption{\oneplusone for Weighted Vertex Cover}
	\begin{algorithmic}[1]
		\STATE The initial solution, $s$, is given: an integer array of size $m$ which used to be a $2$\nobreakdash-approximate solution before changing the weighted graph;
		\WHILE {Stopping criteria not met}
		\STATE{ Set $s':=s$;
			\FORALL {edge $i$ with probability $1/m$}
			\STATE {choose $b\in\{0,1\}$ uniformly at random;
				\IF {$b=0$}\STATE { $s^\prime_i :=s^\prime_i+1$;}\ELSE \STATE {$s^\prime_i := \max\{s^\prime_i-1,0\}$;}\ENDIF}
			\ENDFOR
			\IF {$f(s^\prime)>f(s)$}\STATE {$s:=s^\prime$;}\ENDIF
		}
		\ENDWHILE
	\end{algorithmic}
	\label{alg:WeightedEA}
	
\end{algorithm}

\begin{algorithm}[t]
	\caption{\RLS for Weighted Vertex Cover}
	\begin{algorithmic}[1]
		\STATE The initial solution, $s$, is given: an integer array of size $m$ which used to be a $2$\nobreakdash-approximate solution before changing the weighted graph;
		\WHILE {Stopping criteria not met}
		\STATE{ Set $s':=s$;
			\STATE Choose $i\in\{1,\dots,m\}$ uniformly at random;
			\STATE Choose $b\in\{0,1\}$ uniformly at random;
			\IF {$b=0$}\STATE { $s^\prime_i :=s^\prime_i+1$;}\ELSE \STATE {$s^\prime_i := \max\{s^\prime_i-1,0\}$;}\ENDIF
			
			\IF {$f(s^\prime)>f(s)$}\STATE {$s:=s^\prime$;}\ENDIF
		}
		\ENDWHILE
	\end{algorithmic}
	\label{alg:WeightedRLS}
	
\end{algorithm}

Similar to what we had for the classical vertex cover problem, the fitness function $f$ for the weighted version of vertex cover consists of 3 parts as follows:
\begin{align}\nonumber\
f(s) =& \sum_{i=1}^{m} s_i -\left(\sum_{i=1}^{n}w(v_i)+1\right)\cdot  |\{e\in E\mid  e\cap V_C(s) = \emptyset\}|\\\nonumber
&-(m+1)\cdot\left(\sum_{i=1}^{n}w(v_i)+1\right)\\
&\cdot|\{v\mid\sum_{j\in \{1, \cdots, m\} \mid e_j \cap \{v\} \neq \emptyset}s_j>w(v)\}|\text{.}\label{equ:wfitfunc}
\end{align}
The first part is the sum of weights of edges, which should be maximized. Next there is a penalty for each of the uncovered edges. This part gives the priority to decreasing the number of uncovered edges and lets the algorithm accept a move that decreases the number of uncovered edges, even if the total weight is decreased at the same step. Finally we have a huge penalty for each vertex that violates its constraint. With this amount of the penalty, a solution which has a smaller number of violations is always better that the one with more violations.

In this paper, both dynamic settings that we are considering for the classical vertex cover problem, are investigated for the weighted vertex cover problem as well. Section~\ref{sec:classicalVCP} and Section~\ref{sec:weightedVCP} present the analysis for the classical vertex cover problem and the weighted vertex cover problem, respectively. We perform the runtime analysis with respect to the number of fitness evaluations of the algorithms.

\section{Analysis of the Classical Vertex Cover Problem}
\label{sec:classicalVCP}

In this section, the performance of \EBEAE is studied on the dynamic version of the classical vertex cover problem. In Section~\ref{sec:improveDelete}, we improve the existing results on the re-optimisation time of this algorithm for the situation where a 2\nobreakdash-approximate solution is given and an edge is dynamically removed from the graph. In the second part of this section, Section~\ref{sec:newDynamicSetting}, we analyse the probabilistic dynamic setting for this problem, in which an edge is added to or deleted from the graph at each step of the algorithm with the probability $P_D\leq \frac{1}{55em}$. In the dynamic setting of Section~\ref{sec:improveDelete}, we assume that only one change happens and then we are given a large gap to re-discover a 2-approximate solution. This assumption is relaxed in the dynamic setting of Section~\ref{sec:newDynamicSetting}, where multiple changes may happen before the algorithm finds a new 2-approximate solution.

\subsection{Improving Re-optimisation Time of the \oneplusone for Dynamic Vertex Cover Problem}
\label{sec:improveDelete}
In~\cite{UsDVCGecco2015}, using \oneplusone with the edge-based representation (Algorithm~\ref{alg:EdgeEA}) and the fitness function given in Equation~\eqref{eq:fit_func}, it was shown that if a $2$\nobreakdash-appro-ximate solution is given as the initial solution, after a dynamic deletion happens on the graph, a large number of edges can be uncovered and the re-optimization process takes expected time $O(m \log m)$ to find a $2$\nobreakdash-approximate solution. However, this upper bound is not tight, and is the same as the expected time of finding a $2$\nobreakdash-approximation from an arbitrary solution. In this section, we improve the upper bound on the expected time of re\nobreakdash-optimising $2$\nobreakdash-approximation with this algorithm.

\begin{figure}[t]
	\centering
	\begin{subfigure}[b]{.47\linewidth}
		\centering
		\includegraphics[width=\textwidth]{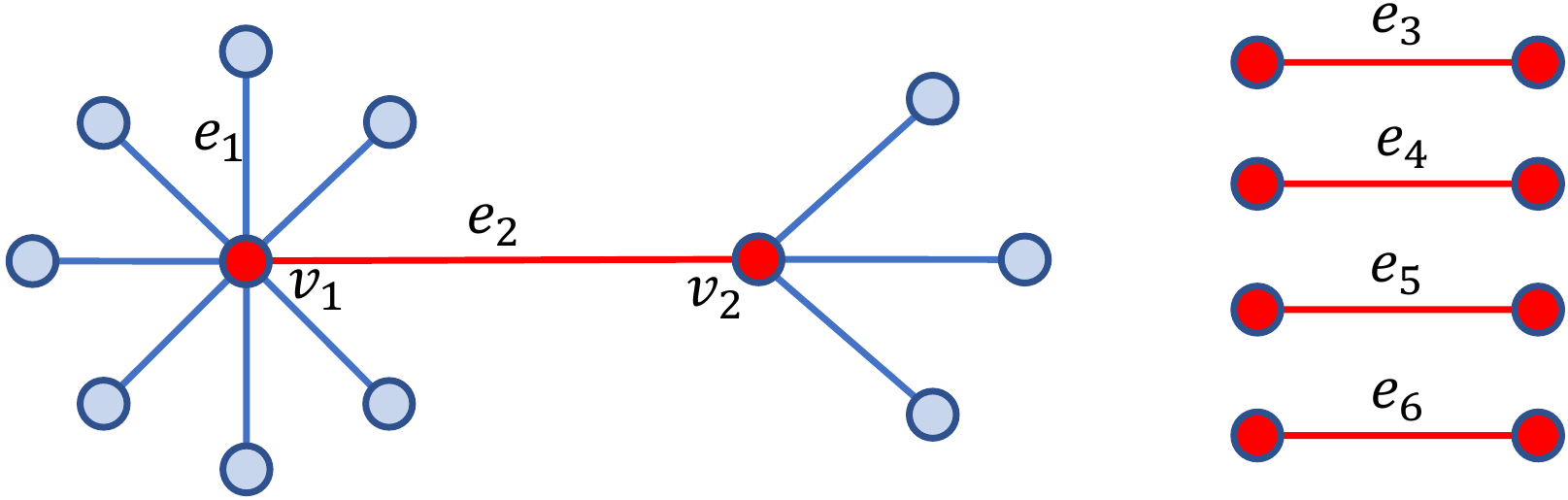}
		\caption{A maximum matching  initial solution.}
		\label{fig:aaa}
	\end{subfigure}
	\hspace{0.5 cm}
	\begin{subfigure}[b]{.47\linewidth}
		\centering
		\includegraphics[width=\textwidth]{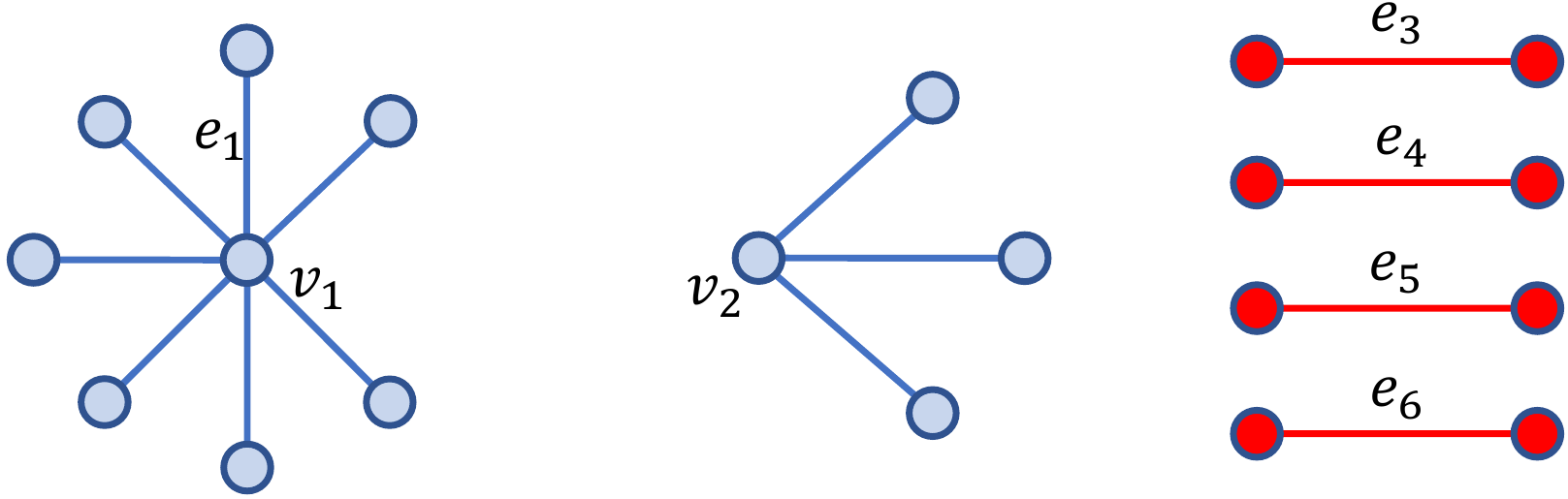}
		\caption{The initial solution after the a dynamic change.}
		\label{fig:bbb}
	\end{subfigure}
	\begin{subfigure}[b]{.47\linewidth}
		\centering
		\includegraphics[width=\textwidth]{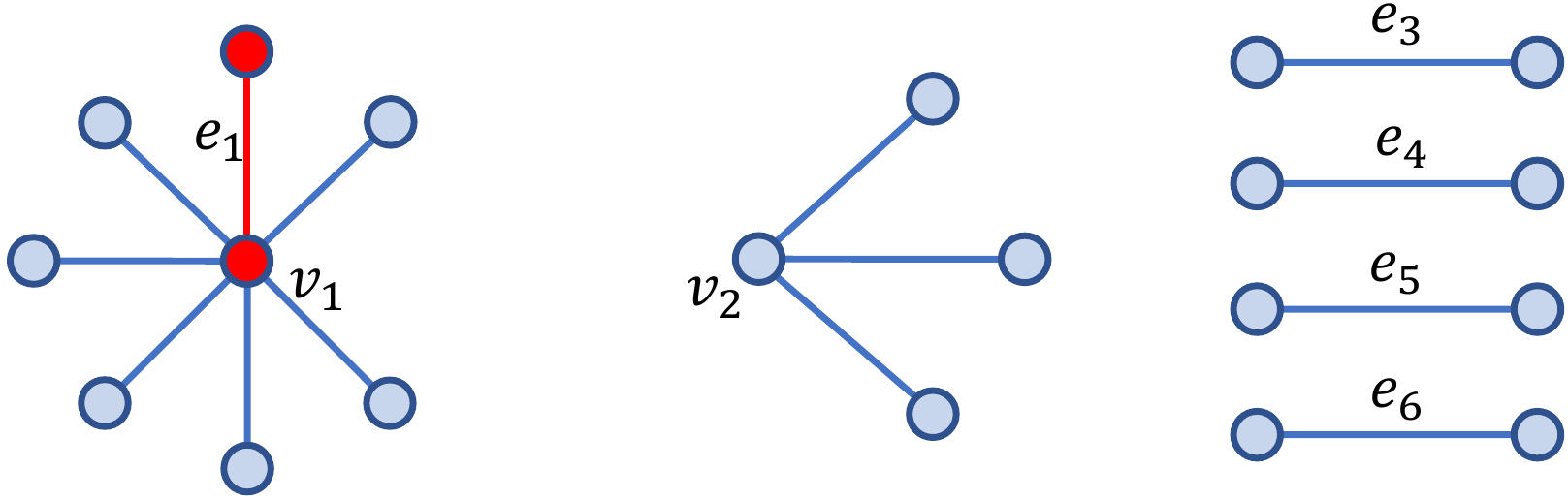}
		\caption{Multiple bit-flips on solution \ref{fig:bbb} done by \EBEAE which is accepted according to the fitness function.}
		\label{fig:ccc}
	\end{subfigure}
	\caption{A sample initial solution for \EBEAE and the damage caused by the multiple bit-flips. Red edges demonstrate the chosen edges in the current solution and red vertices are the corresponding cover set.}
	\label{fig:multi-bit-flip}
\end{figure}

Figure \ref{fig:multi-bit-flip} depicts the main challenge in the analysis of \EBEAE when an edge is dynamically deleted from the graph. The set of all nodes of edges $\{e_2,\cdots,e_6\}$ in Figure \ref{fig:aaa} is a 2-approximate solution for the minimum vertex cover problem. Let the dynamic change delete edge $e_2$. This move uncovers all edges that are connected to $v_1$ and $v_2$ (Figure \ref{fig:bbb}). In order to cover all these uncovered edges, the algorithm needs to pick only two new edges that add $v_1$ and $v_2$ to the cover set. However,  \EBEAE can perform a multiple bit-flip that makes the situation more complicated. For example, at the same time that $e_1$ is added to the solution, edges $\{e_3,\cdots,e_6\}$ can be removed from the solution (Figure \ref{fig:ccc}). Although this solution has fewer uncovered edges and will be accepted by the algorithm, it is more difficult to achieve a maximum matching from this solution and the algorithm needs to do at least 4 bit-flips to re-optimize the problem.

We divide the analysis into two phases, where both take expected $O(m)$ steps. In the first phase we prove that the number of uncovered edges is decreased to a constant, and in the second phase, we show that all the edges become covered.

Consider a solution $s$ that is a matching but not a maximal matching. The cover set, $V_{C}(s)$, derived from this solution is not a complete cover. Let $C$ be the minimal set of vertices that are required to be added to $V_{C}(s)$ to make it a complete cover ($C=\{v_1,v_2\}$ in Figure \ref{fig:bbb}). Initially, this set consist of at most two nodes (both nodes of the deleted edge), but during the run of the \oneplusone, when some nodes are removed from this set, new nodes can be added to it, since more than one mutation can happen at the same step.

We define the set of nodes $C_1$ as the following. Initially, let $C_1\subseteq C$ consist of all nodes of $C$ that are connected to more than 5 uncovered edges (\eg $C_1 = \{v_1\}$ in Figure \ref{fig:bbb}). Since the number of nodes in $C$ is at most 2 at the beginning of the process, the initial number of nodes in $C_1$ is also bounded by 2. During the process of the algorithm, more nodes with this property are added to $C$, but we only add them to $C_1$ if at the same step, at least one other node from $C_1$ is included in the new solution and removed from $C_1$. 

In the analysis of the first phase, using the drift on the number of nodes in $C_1$, we show that this set becomes empty in $O(m)$. After this point, no nodes can be added to $C_1$, due to definition of $C_1$. Let $E_u$ be the subset of uncovered edges that do not have a node in $C_1$. We prove that at the end of the first phase, $E_u$ consists of a constant number of edges, and using the drift analysis on $|E_u|$, we show that all edges are covered in $O(m)$. 

Let ${C_1}^t$ and $E(\Delta_{C_1}^t)=E\left(|C_1^t|- |C_1^{t+1}|\mid |C_1^t|\right) $ denote $C_1$ at step $t$ of the run of the algorithm, and the drift on the size of this set, respectively. In order to find $E(\Delta_{C_1}^t)$, we first introduce a partitioning on the selected edges and prove a property (Lemma~\ref{lem:selectedSetsUpperBound}) about this partitioning and the number of uncovered edges.

Let $E_i(s), 1\leq i\leq m,$ be the set of selected edges in solution $s$, that deselecting each of them uncovers $i$ covered edges. Moreover, each covered edge of the graph is either covered by one node or two nodes of the induced node set of $s$. Let the set of edges that are covered from both ends be $D(s)$, \ie $D(s)= \left\{e=\{v, u\}\mid v \in V(s) \wedge u \in V(s) \right\}$. 
According to the definitions of $D(s)$ and $E_i(s)$ and the total number of covered edges, the following lemma gives us an equation that helps us in the proof of Lemma \ref{lem:driftOnC1}.
\begin{lemma}
	For any solution $s$, $|D(s)| +\sum_{i=1}^{m} i\cdot |E_i(s)| \leq m-k$, where $k$ is the number of uncovered edges of solution $s$ and $m$ is the total number of edges. 
	\label{lem:selectedSetsUpperBound}
\end{lemma}
\begin{proof}
	Let us first consider all covered edges except those that are in $D(s)$. By definition of $E_i(s)$, $1\leq i\leq m$, deselecting each edge of $E_i(s)$ uncovers $i$ edges. This implies that all of these $i$ edges are only covered by the deselected edge and none of them is uncovered by deselecting another edge. Therefore, each covered edge that is not in $D(s)$, is counted at most once in $\sum_{i=1}^{m} i\cdot |E_i(s)|$. 
	
	On the other hand, by definition of $D(s)$, none of the edges of $D(s)$ are uncovered when one of the edges of $E_i(s)$, $1\leq i\leq m$ is deselected. Therefore, edges of $D(s)$ are not counted in $\sum_{i=1}^{m} i\cdot |E_i(s)|$. Moreover, the number of covered edges is $m-k$, which completes the proof.
\end{proof}

Using Lemma~\ref{lem:selectedSetsUpperBound} in the following lemma we find a lower bound on the value of $E(\Delta_{C_1}^t)$.

\begin{lemma}
	\label{lem:driftOnC1}
	At each step of \EBEAE, $E(\Delta_{C_1}^t)\geq\frac{6-2e}{em}\cdot |C_1| $.
\end{lemma}

\begin{proof}
	By definition of $C_1$, changes to this set can only happen at the steps where at least one node of this set is included in the solution. Moreover, a node of $C_1$ will be included in the solution if exactly one of its adjacent uncovered edges is selected, which happens with a probability of at least $P_{C_1}\geq \frac{6}{em} |C_1|$ at each step. In the proof of this lemma, we filter the steps and only consider the steps in which at least one node of $C_1$ is included, and show that the expected change on $|C_1|$ in those steps is at least $(1-\frac{2e}{6})$; therefore, $E(\Delta_{C_1}^t)\geq P_{C_1}\left(1-\frac{2e}{6}\right) \geq \frac{6-2e}{em}\cdot |C_1| $.
	
	From this point of the proof, we filter the steps and only consider the steps in which at least one node of $C_1$ is included. Let the drift on $|C_1|$ in these steps be denoted by $E_f(\Delta_{C_1}^t)$. We aim to find a lower bound on $E_f(\Delta_{C_1}^t)$.
	
	Let $P_{Acc}$ denote the probability that  the whole move of a step is accepted.
	When one node of $C_1$ is included and no other mutations happen at the same step, the move is accepted by the algorithm. Therefore, 
	
	\begin{eqnarray}
	\label{eq:P_Accept}
	P_{Acc}\geq  \left(1-\frac{1}{m}\right)^{m-1}\geq \frac{1}{e}\text{.}
	\end{eqnarray}
	
	Moreover, let $Acc$ denote the event that the whole move in a step is accepted. Also let $P(\text{bit}\mid Acc)$ denote the probability of mutating an edge $e$, under the condition that event $Acc$ has occurred. By the definition of conditional probability, we know that 
	$$P(\text{bit}\mid Acc)= \frac{P(\text{bit} \cap P_{Acc})}{P_{Acc}}\leq \frac{P(\text{bit})}{P_{Acc}}\text{.}$$
	Using Equation~\eqref{eq:P_Accept}, we get:
	$$P(\text{bit}\mid Acc)\leq eP(\text{bit} )\text{,}$$
	where $P(\text{bit} )$ is the unconditional probability of flipping $e$, which is $\frac{1}{m}$. This implies that
	\begin{eqnarray}
	\label{eq:Pbit}
	P(\text{bit}\mid Acc)\leq \frac{e}{m}\text{.}
	\end{eqnarray}
	
	The drift on $|C_1|$ can be presented as $$E_f(\Delta_{C_1}^t)= E_f^+(\Delta_{C_1}^t) -E_f^-(\Delta_{C_1}^t)\text{,}$$ 
	where $E_f^+(\Delta_{C_1}^t)$ is the expected number of nodes that are removed from $C_1$ at each step, and $E_f^-(\Delta_{C_1}^t)$ is the expected number of nodes that are added to $C_1$ at each step.  Since we are only considering the steps in which at least one node of $C_1$ is included, we have
	$$E_f^+(\Delta_{C_1}^t)\geq1.$$
	
	Here we find an upper bound on $E_f^-(\Delta_{C_1}^t)$. Nodes can only be added to $C_1$, when a selected edge that covers more than 6 edges is deselected. 
	We need to find the expected number of mutating edges of type $E_i(s)$, $i\geq 6$. Since deselecting each of these edges can add at most 2 nodes to $C_1$, $E_f^-(\Delta_{C_1}^t)$ is  upper bounded by:
	$$E_f^-(\Delta_{C_1}^t)\leq 2\sum_{i=6}^{\infty} |E_i(s)|\cdot P(\text{bit}\mid Acc) \text{.}$$
	From Equation~\eqref{eq:Pbit}, we get:
	$$E_f^-(\Delta_{C_1}^t)\leq 2\sum_{i=6}^{\infty} |E_i(s)|\cdot \frac{e}{m} \leq\frac{1}{m}\sum_{i=6}^{\infty} \frac{2e i}{i}\cdot |E_i(s)| \leq\frac{2e}{6m}\sum_{i=6}^{\infty} i\cdot |E_i(s)|\text{.} $$
	
	On the other hand, Lemma~\ref{lem:selectedSetsUpperBound} implies that 
	$\sum_{i=6}^{m} i\cdot |E_i(s)| \leq m-k$, which gives us:
	\begin{eqnarray}
	\label{eq:negativeDriftOnFilteredSteps}
	E_f^-(\Delta_{C_1}^t) \leq \frac{2e(m-k)}{6m}\leq \frac{2e}{6}\text{.}
	\end{eqnarray}
	Therefore, the drift on $|C_1|$ is 
	\begin{eqnarray}
	\label{eq:driftOnFilteredSteps}
	E_f(\Delta_{C_1}^t)= E_f^+(\Delta_{C_1}^t) -E_f^-(\Delta_{C_1}^t)\geq 1-\frac{2e}{6}\text{,}
	\end{eqnarray}
	which completes the proof.
\end{proof}

In the following lemma, we prove that the set $C_1$ becomes empty in expected time $O(m)$. Moreover, in Lemmata~\ref{lem:TotalDecreaseThanCanHappenOnC1} to~\ref{lem:sizeOfEuAfterPhase1}, we prove that the total number of uncovered edges at the beginning of the second phase is a constant. Then in Lemma~\ref{lem:driftOnEu} we find the drift on $|E_u|$ during the second phase, which helps us with the proof of Theorem~\ref{thm:improvedDelete}.

\begin{lemma}
	\label{lem:C1Empty}
	Starting with a situation where $|C_1|= c$, $c\geq 0$, the expected time until the algorithm reaches a situation where $|C_1|=0$ is at most $\frac{em(1+\ln (c))}{6-2e}$.
\end{lemma}
\begin{proof}
	According to Lemma~\ref{lem:driftOnC1}, the drift on $C_1$ is at least $\frac{6-2e}{em}\cdot |C_1|  $. Therefore, since the algorithm starts with $|C_1|\leq c$ and the minimum value of $|C_1|$ before reaching $|C_1|=0$ is 1, by multiplicative drift analysis, we find the expected time of at most
	$$\frac{1+\ln (c)}{\frac{6-2e}{em}}= \frac{em\left(1+\ln (c)\right)}{6-2e}$$ 
	to reach a solution $s$ where $|C_1|=0$.
\end{proof}

\begin{lemma}
	\label{lem:TotalDecreaseThanCanHappenOnC1}
	Starting with $|C_1|=c$, $c\geq 0$ a constant integer, the expected total number of steps at which a node can be removed from $C_1$ is upper bounded by $\frac{c}{1-2e/6}$.
	%$c+\frac{2ec}{6-2e}$.
\end{lemma}
\begin{proof}
	Similar to the proof of Lemma~\ref{lem:driftOnC1}, for the proof of this lemma we filter the steps and only consider the steps at which at least one node of $C_1$ is included, because no change on $C_1$ can happen in all other steps.
	
	In the proof of Lemma~\ref{lem:driftOnC1}, we proved in Equation~\eqref{eq:driftOnFilteredSteps} that the drift on $|C_1|$ on the filtered steps is
	$E_f(\Delta_{C_1}^t)\geq 1-\frac{2e}{6}$.
	Using additive drift analysis and the assumption that $|C_1|=c$ at the start of the process, we can conclude that we reach $|C_1|=0$ in expected $\frac{c}{1-2e/6}$ filtered steps.
	% \textbf{Maybe this is enough?!}
	%
	%Let $X$ denote the total number of steps at which a node can be removed from $C_1$ and . The expected value of $X$
	%Moreover, in the proof of Lemma~\ref{lem:driftOnC1}, we proved in Equation~\eqref{eq:negativeDriftOnFilteredSteps} that the expected number of nodes that are added to $|C_1|$ in an accepted filtered step is upper bounded by
	%$E_f^-(\Delta_{C_1}^t) \leq  \frac{2e}{6}$.
	%This implies that in expectation at most $\left(\frac{c}{1-2e/6}\right) \left(\frac{2e}{6}\right)= \frac{2ec}{6-2e}$ nodes can be added to $C_1$ during the process of the algorithm.
	%
	%Together with the initial number of nodes in $C_1$, the expected total number of steps at which a node can be removed from $C_1$ is upper bounded by
	%$c+\frac{2ec}{6-2e}$.
\end{proof}

\begin{lemma}
	\label{lem:incrementsOnEuDuringPhase1}
	Starting with $|C|=c$, $c\geq 0$ a constant integer, the expected  number of edges that can be added to $|E_u|$ by the end of the first phase is upper bounded by $\frac{ce}{1-2e/6}$.
\end{lemma}
\begin{proof}
	During the process of the algorithm, at the steps where $C_1$ does not face a change, a change on the total number of uncovered edges can only happen through $E_u$. Therefore, at these steps, to have an accepted move, the number of edges of $E_u$ can not increase. The reason is that the fitness function is defined in such a way that increasing the total number of uncovered edges is not accepted. Hence, in order to find the increments on the number of edges of $E_u$, we only need to consider the steps in which at least one node from $C_1$ is included. We apply the same filtering on the steps that we had for the proof of Lemma~\ref{lem:driftOnC1} and find the expected number of edges that are added to $E_u$ in those steps.
	
	In Equation~\eqref{eq:Pbit} of the proof of Lemma~\ref{lem:driftOnC1}, we found the minimum probability of flipping an edge, under the condition that the move is accepted. Based on this probability, for each $0<i<m$ the expected number of edges that are deselected from $E_i(s)$ at each filtered step is $\frac{e}{m} |E_i(s)|$. Since each of them uncover $i$ edges, the expected number of uncovered edges will be $\frac{e}{m}\sum_{i=1}^m i|E_i(s)| $. On the other hand, Lemma~\ref{lem:selectedSetsUpperBound} implies that $\sum_{i=1}^m i|E_i(s)| \leq m-k$. Therefore, we conclude that  the expected number of uncovered edges that are added to $E_u$ at each filtered step is at most $\frac{e}{m} (m-k) \leq e$.
	
	Moreover, according to Lemma~\ref{lem:TotalDecreaseThanCanHappenOnC1}, the expected number of steps at which a node from $C_1$ can be included in the solution is upper bounded by $\frac{c}{1-2e/6}$. This implies that the expected increase on $|E_u|$ by the end of the first phase is upper bounded by 
	$$\left(\frac{c}{1-2e/6}\right)\cdot e = \frac{ce}{1-2e/6}\text{.}$$
\end{proof}

\begin{lemma}
	\label{lem:sizeOfEuAfterPhase1}
	Starting with $|C|=c$, $c\geq 0$ a constant integer, the expected  number of edges in $|E_u|$ by the end of the first phase is upper bounded by $5c +\frac{(4e-e^2)c}{3-e}$.
\end{lemma}
\begin{proof}
	By definition of $C$ and $C_1$, at the beginning of the process, all uncovered edges that do not have a node in $C_1$, must have a node in $C \setminus C_1$. The number of these nodes is upper bounded by $c$, because $|C|= c$. Also, the number of uncovered edges that are adjacent to each of them is at most $5$. Therefore,  at the start of the process, $|E_u|\leq 5c$.
	
	Moreover, according to Lemma~\ref{lem:incrementsOnEuDuringPhase1}, starting with $|C|=c$, $c\geq 0$ a constant integer, the expected  number of edges that can be added to $|E_u|$ by the end of the first phase is upper bounded by $\frac{ce}{1-2e/6}$. 
	Together with the initial number of uncovered edges in $E_u$, the total number of edges in $E_u$ is in expectation upper bounded by $5c +\frac{ce}{1-2e/6}$.
\end{proof}

Denoting  $E_u$ at step $t$ of the algorithm by $E_u^t$ and the drift on $|E_u|$ at that step by $E(\Delta^t_{E_u})=E(|E_u^t|-|E_u^{t+1}|\mid |E_u^t|) $, the following lemma proves a lower bound on  $E(\Delta^t_{E_u})$, during the second phase of the analysis.

\begin{lemma}
	\label{lem:driftOnEu}
	At a step $t$ of \EBEAE after reaching $|C_1|=0$, the drift on $|E_u|$ is $E(\Delta^t_{E_u})\geq \frac{|E_u^t|}{em}$.
\end{lemma}

\begin{proof}
	After reaching $|C_1|=0$, the edges of $E_u$ are the only uncovered edges of the solution. Therefore, due to the definition of the fitness function, $|E_u|$ never increases during the run of the algorithm. 
	
	Moreover, selecting one edge of $E_u$ reduces the number of uncovered edges by at least one, and the move is accepted if no other mutations happen at the same step, which happens with probability $\frac{1}{m}\left(1-\frac{1}{m}\right)^{m-1}\geq \frac{1}{em}$.  There are $|E_u|$ edges in this set, resulting in $|E_u^t|$  mutually exclusive events of improving single mutation moves at each step. Therefore we can conclude that $E(\Delta_{E_u})\geq \frac{|E_u^t|}{em}$.
\end{proof}

We now prove the main theorem of this section. A dynamic change affects the graph by either deleting an edge or adding it. However, it is already shown that \EBEAE restores the quality of $2$\nobreakdash-approximation when a new edge is added dynamically in expected time $O(m)$~\cite{UsDVCGecco2015}. In Theorem \ref{thm:improvedDelete} we prove that the expected re-optimisation time of \EBEAE after a dynamic deletion is also $O(m)$. 

\begin{theorem}
	Starting with a $2$\nobreakdash-approximate solution $s$, which is a maximal matching, \EBEAE rediscovers a $2$\nobreakdash-approximation when one edge is dynamically deleted from the graph in expected time $O(m)$.
	\label{thm:improvedDelete}
\end{theorem}

\begin{proof}
	Let $e=\{v_1, v_2\}$ be the edge that is deleted from the graph. If $e \notin E(s)$ then $s$ is still a maximal matching and corresponds to a $2$\nobreakdash-approximate vertex cover. If $e \in E(s)$, then it is deleted from the solution as well. The new $s$ is still a matching but may not be a maximal matching. By Lemma~22 of~\cite{UsDVCGecco2015}, we know that a non-matching solution is never accepted by the algorithm; therefore, we only need to find the expected time to reach a solution with no uncovered edges and know that it is a maximal matching, which induces a $2$\nobreakdash-approximate solution. 
	
	The number of uncovered edges of $s$ after the dynamic deletion can be in $\Omega(m)$, but all of them can be covered by including the two nodes of the deleted edge; therefore, $C\leq 2$ holds just after the dynamic change. 
	Moreover, according to Lemma~\ref{lem:C1Empty}, in expected time $\frac{em(1+\ln (2))}{6-2e}= O(m)$ the first phase ends, as the algorithm reaches a situation where $|C_1|=0$. 
	
	At the beginning of the second phase, the set $E_u$, which includes all the uncovered edges of the current solution, can have a size between 0 and $m$,  but due to Lemma~\ref{lem:sizeOfEuAfterPhase1}, we know that the expected size of this set is $E[|E_u|]\leq 10 +\frac{2e}{1-2e/6} $. If we denote by $T$, the required time until reaching $|E_u|=0$, then by the law of total expectation we have
	$$E[T]=E\left[E[T\mid |E_u|]\right]\text{.}$$
	
	Furthermore, according to Lemma~\ref{lem:driftOnEu}, during the second phase of the analysis, we have $E(\Delta^t_{E_u})\geq \frac{|E_u^t|}{em}$. Moreover, the minimum value of $|E_u|$ before reaching $|E_u|=0$ is 1. Therefore, using multiplicative drift analysis, we have 
	$$E[T\mid |E_u|]\leq \frac{1+\ln (|E_u|)}{\frac{1}{em}}= em + em\ln(|E_u|)\text{,}$$
	which implies 
	$$E[T] \leq E\left[em + em\ln(|E_u|)\right] = em + em\cdot E\left[\ln(|E_u|)\right]\text{.}$$
	Now, by applying Jensen's Inequality, we find that $E\left[\ln(|E_u|)\right]\leq \ln\left(E[|E_u|]\right)$, which together with the above inequality implies
	$$E[T]\leq em + em\cdot \ln(E[|E_u|]) \leq em + em\cdot \ln\left(10 +\frac{2e}{1-2e/6}\right)= O(m)\text{.}$$
	The last inequality holds due to Lemma~\ref{lem:sizeOfEuAfterPhase1}.
	
	Altogether,  since both phases of our analysis until finding a 2-approximate solution take $O(m)$, the theorem is proved.
\end{proof}

\subsection{Complexity Analysis for the Dynamic Vertex Cover Problem With Probabilistic Dynamic Changes }\label{sec:newDynamicSetting}

In this section, we consider the probabilistic setting for the dynamic vertex cover problem, in which a dynamic change happens on the graph at each step of the algorithm with probability $P_D$. Similar to the previous section, we assume that the maximum number of edges in the graph is $m$. The analysis of this section shows that when $P_D$ is sufficiently small, \oneplusone can find a 2-approximate solution in expected polynomial time. Moreover, we show that if a maximal matching solution is provided before the first dynamic change, \oneplusone can re-discover a 2-approximate solution in expected linear time. Assuming that $P_D\leq\frac{1}{2000em}$, in the first theorem of this section we show that our upper bound holds. 
% We assume $P_D\leq\frac{1}{2000em}$ in this analysis and, then in the first theorem of this section; we find that our upper bound holds.

We use the same definition of $C_1$ and $E_u$ that we had in section~\ref{sec:improveDelete}, except that when a dynamic change happens new nodes that can cover more than 5 edges are added to $C_1$, if they are adjacent to the  edge that has been dynamically deleted. If we did not have dynamic changes, the expected changes of $C_1$ and $E_u$ were the same as the previous section. Observe that with this definition, each dynamic change can add at most 2 nodes to $C_1$, and 10 edges to $E_u$. We start with a couple of lemmata that show the drift on $C_1$ and an upper bound on the expected time until $|C_1 |=0$. Then in Lemmata~\ref{lem:NumberOfStepsWhereC1Reduces} and~\ref{lem:IncreaseOnEu} we find the expected increase that happens on $E_u$ during the process of the algorithm until reaching $|C_1|=0$. Moreover, in Lemmata~\ref{lem:driftOnEuDuringSmallPhases} and~\ref{lem:driftConstantOnEuDuringSmallPhases}, we find the expected change that happens on $|E_u|$ during a phase in which $|C_1|=0$ holds. Finally, using all these lemmata, we prove the main results of this section in Theorems~\ref{thm:dyn_arbitrary} and~\ref{thm:dyn_2approx}.

\begin{lemma}
	\label{lem:Setting2DriftOnC1}
	If $P_D\leq\frac{1}{2000em}$, unless we reach a situation where $| C_1 |=0$, at each step of \oneplusone, $E(\Delta_{C_1}^t)\geq\frac{1996|C_1|}{4000em}$.
\end{lemma}

\begin{proof}
	The drift on the size of $C_1$ in the dynamic setting that we are analysing in this section consists of the expected changes that the \oneplusone makes on $C_1$ in addition to the expected changes that are caused by the dynamic changes of the graph. We denote the latter by $E(\Delta_{C_1}^D)$. Lemma \ref{lem:driftOnC1} gives us the drift on $|C_1^t|$ obtained by the \oneplusone with the value of at least $\frac{6-2e}{em}\cdot|C_1|$. Therefore, for the total drift on $|C_1^t|$ we have $$E(\Delta_{C_1}^t)\geq\frac{6-2e}{em}\cdot|C_1| +E(\Delta_{C_1}^D)\text{.}$$
	Each dynamic change adds at most two new nodes to $C_1$. Moreover, a dynamic change takes place in each step with the probability $P_D\leq\frac{1}{2000em}$. Thus, the expected increase on $|C_1|$ caused by a dynamic change in each step is  at most $\frac{-2}{2000em}$ and we have 
	$$E(\Delta_{C_1}^t)\geq\frac{6-2e}{em}\cdot|C_1| +\frac{-2}{2000em}\text{.}$$
	Knowing that $C_1\ne\emptyset$ we find 
	$$E(\Delta_{C_1}^t)\geq\frac{2000(6-2e)|C_1|-2}{2000em}\geq\frac{1996|C_1|}{4000em} \text{.}$$
\end{proof}

\begin{lemma}
	\label{lem:setting2C1Empty}
	Starting with a situation where $| C_1 |= c$, $c\geq 0$, the expected time until the algorithm reaches a situation where $| C_1 |=0$ is upper bounded by $\frac{4000em(1+\ln (c))}{1996} $. Furthermore, with probability at least $e^{-r}$, for $r>0$, the required time until reaching $| C_1 |=0$ is upper bounded by $\frac{4000em(r+\ln (c))}{1996} $.
\end{lemma}
\begin{proof}
	With a similar argument as we had in Lemma~\ref{lem:C1Empty}, by means of Lemma~\ref{lem:Setting2DriftOnC1} and multiplicative drift analysis, we find the expected time of at most
	$$\frac{1+\ln (c)}{\frac{1996}{4000em}}= \frac{4000em(1+\ln (c))}{1996} = O(m \ln(c))$$ 
	to reach a solution $s$ where $| C_1 |=0$.
	Moreover, by multiplicative drift tail bounds, we can conclude that with probability at least $e^{-r}$ the number of required steps until reaching the desired solution is at most
	$$\frac{r+\ln (c)}{\frac{1996}{4000em}}= \frac{4000em(r+\ln (c))}{1996}\text{.}$$
\end{proof}

\begin{lemma}
	\label{lem:NumberOfStepsWhereC1Reduces}
	Starting with a situation where $| C_1 |= 2$, until the algorithm reaches $| C_1 |=0$, the expected number of steps in which $|C_1|$ changes,  is at most $22+ \frac{77}{1996}$.
\end{lemma}
\begin{proof}
	By definition of $C_1$, changes on this set can only happen at the steps where either a dynamic change happens or at least one node of this set is included in the solution. From Equation~\ref{eq:driftOnFilteredSteps}, we know that the expected change that happens on $|C_1|$ at the steps where at least one node of this set is included in the solution, is at least $1-2e/6$. Here, we first find the expected increase that dynamic changes can cause on $|C_1|$, which together with the initial value of $|C_1|$ gives us the expected total decrease that should happen on $|C_1|$ to reach 0. Then we use the constant drift of Equation~\ref{eq:driftOnFilteredSteps}, to find the expected number of required steps of that kind.
	
	According to Lemma~\ref{lem:setting2C1Empty} in expectation it takes 
	$\frac{4000em(1+\ln (2))}{1996}$ steps to reach $|C_1|=0$, during which, in expectation $P_D\cdot\left(\frac{4000em(1+\ln (2))}{1996}\right)$ dynamic changes happen. Since $P_D\leq\frac{1}{2000em}$, the expected number of dynamic changes in this phase is $\frac{2(1+\ln (2))}{1996}$. Each dynamic change increases $|C_1|$ by at most 2. Therefore, together with the initial value of $C_1$ the expected total decrease that needs to happen on $|C_1|$ to reach $|C_1|=0$ is $2+\frac{4(1+\ln (2))}{1996}$.
	
	Now we only need to count the number of steps in which at least one node of $|C_1|$ is added to the solution. Since at each of these steps $|C_1|$ is reduced by an expected value of $1-2e/6$, 
	denoting the expected number of these steps until reaching $|C_1|=0$ by $E[T_{C_1}]$, and the total required decrease by $C_{1_D}$, we have
	$$E[T_{C_1}]= E\left[ E[T_{C_1}\mid C_{1_D}]\right] = E\left[  \frac{C_{1_D}}{1-\frac{2e}{6}}\right]\text{.}$$
	By linearity of expectation, we have
	$$E[T_{C_1}]=  \frac{E[C_{1_D}]}{1-\frac{2e}{6}} \leq \frac{2+\frac{4(1+\ln (2))}{1996}}{1-\frac{2e}{6}}<22+\frac{73}{1996}\text{.}$$
	Together with the expected number of steps where a dynamic change happens, the expected number of steps in which $|C_1|$ changes,  is upper bounded by
	$$22+\frac{73}{1996} +\frac{2(1+\ln (2))}{1996}< 22+ \frac{77}{1996}\text{.}$$
\end{proof}

Let ${E_u}^t$ and $C_1^t$ denote $E_u$ and $C_1$ at step $t$ of the run of the algorithm, respectively.
Moreover, let $ E[\Delta_{E_u}^{tt'}]=E\left[|E_u^{t'}|- |E_u^{t}| \mid | C_1^t |= 2 \wedge | C_1^{t'} |= 0 \right] $ denote the expected change that happens on the size of this set from step $t$ to $t'$, when $| C_1^t |= 2$ and $| C_1^{t'} |= 0$. In the following lemma, we find an upper bound on $E(\Delta_{E_u}^{tt'})$.

\begin{lemma}
	\label{lem:IncreaseOnEu}
	Consider steps $t$ and $t'$ where $| C_1^t |= 2$ and $| C_1^{t'} |=0$. We have $E(\Delta_{E_u}^{tt'})\leq 600+ \frac{2130}{1996} $.
\end{lemma}
\begin{proof}
	During the process of the algorithm, at the steps where neither a dynamic change happens nor $C_1$  changes, to have an accepted move, the number of edges of $E_u$ cannot increase. Therefore, we only find the expected increase on $|E_u|$ during the dynamic changes and also during the steps where $C_1$ is changed.
	
	Let $X$ be a random variable denoting the number of dynamic changes that happen before reaching $C_1=0$. The expected increase of $|E_u|$ as a result of these changes, denoted by $Y$, is $E[Y|X=x]\leq 10x$, because each dynamic change increases $|E_u|$ by at most 10.  Similar to the proof of Lemma~\ref{lem:NumberOfStepsWhereC1Reduces}, we can show that the expected value of $X$ is at most $\frac{2(1+\ln (2))}{1996}$. By the law of total expectation and by linearity of expectation we find the expected value of $Y$ as $E[Y]= E[E[Y|X]]= E[10X]= 10E[X]\leq  \frac{20(1+\ln (2))}{1996}<\frac{34}{1996}\text{.}$ 
	
	Let $X'$ be the number of steps at which $C_1$  changes before reaching $C_1=0$, and $X''$ be the number of edges that are deselected from the solution at these steps. In Equation~\eqref{eq:Pbit} of the proof of Lemma~\ref{lem:driftOnC1}, we found an upper bound on the probability of flipping an edge, under the condition that the move is accepted. Based on this probability, the expected number of edges that are deselected from the solution at each step is at most $e$. Therefore, the expected number of these edges at $X'=x'$ steps is $E[X''|X'=x']\leq ex'$. Moreover, by Lemma~\ref{lem:NumberOfStepsWhereC1Reduces} we have $E[X']\leq 22+ \frac{77}{1996}$. By the law of total expectation, we find $E[X'']=E[E[X''|X']]\leq E[eX']= eE[X']\leq 22e+ \frac{77e}{1996}$.

	Since each of these  edges, when deselected, uncover at most $10$ edges, at most $10X''$ edges are added to $E_u$ at these steps. Denoting the increase of $|E_u|$ as a result of these changes by $Z$, by the law of total expectation we find the expected value of $Z$ as $E[Z]= E[E[Z|X'']]\leq E[10X'']=10E[X'']= 220e+ \frac{770e}{1996}$.

	Together with the increase that happens at the steps of dynamic changes, the expected increase on $|E_u|$ until reaching $|C_1|=0$ is upper bounded by
	$$220e+ \frac{770e}{1996}+ \frac{34}{1996}<600+\frac{2130}{1996}\text{.}$$
\end{proof}

\begin{lemma}
	Consider a phase of $X\geq em$ steps in which dynamic changes do not happen and $|C_1|=0$ holds. Let $|E_u^{str}|$ and $|E_u^{end}|$ be the value of $|E_u|$ at the start and end of this phase, respectively. We have 
	$$E\left[|E_u^{str}|-|E_u^{end}|\mid |E_u^{str}|\right]\geq  (1-1/e)|E_u^{str}|\text{.}$$
	\label{lem:driftOnEuDuringSmallPhases}
\end{lemma}
\begin{proof}
	Since we have assumed that $C_1$ is empty in this phase, we can use Lemma~\ref{lem:driftOnEu} and conclude that after one step, the expected value of $|E_u|$ is at most $(1-\frac{1}{em})|E_u^{str}|$. Similarly, after each step, a coefficient of $1-\frac{1}{em}$ is multiplied to this value, which implies that after $X$ steps, the expected value of $|E_u|$ is 
	$$E\left[|E_u|\mid |E_u^{str}|\right]\leq \left(1-\frac{1}{em}\right)^X|E_u^{str}|\text{.}$$
	Moreover, we have assumed that $X\geq  em$ holds. Since $(1-1/em)^{em}<1/e$, the right hand side of the inequality above is simplified as 
	%$$E\left[|E_u^{end}|\mid |E_u^{str}|\right]\leq \frac{1}{2}(1-\frac{1}{em})^{em/2}|E_u^{str}| + \frac{1}{2}|E_u^{str}|.$$
	$$E\left[|E_u|\mid |E_u^{str}|\right]\leq \frac{1}{e}|E_u^{str}|\text{,}$$
	which implies 
	$$E\left[|E_u^{str}|-|E_u^{end}|\mid |E_u^{str}|\right]\geq  (1-1/e)|E_u^{str}|\text{.}$$
\end{proof}

\begin{lemma}
	Consider a phase of $X$ steps, $X$ a random variable, in which dynamic changes do not happen and $|C_1|=0$ holds. Let $|E_u^{str}|$ and $|E_u^{end}|$ be the value of $|E_u|$ at the start and end of this phase, respectively. Assuming that $E[X] \geq cem$, $c\geq 0$, we either reach $|E_u|=0$ before the end of this phase, or we have 
	$$E\left[|E_u^{str}|-|E_u^{end}|\mid |E_u^{str}| \right]\geq  c\text{.}$$
	\label{lem:driftConstantOnEuDuringSmallPhases}
\end{lemma}
\begin{proof}
	Similar to the proof of Lemma~\ref{lem:driftOnEuDuringSmallPhases}, we can use Lemma~\ref{lem:driftOnEu}, which gives us the drift on $|E_u|$ at a step $t$ as $E[\Delta^t_{E_u}]\geq \frac{|E_u^t|}{em}$. If we reach  $|E_u|=0$ before the end of this phase, then the statement is proved. Otherwise, at each step $t$ of the considered phase we have $E[\Delta^t_{E_u}]\geq \frac{1}{em}$. Therefore, by the law of total expectation we have
	$$E\left[|E_u^{str}|-|E_u^{end}| \mid |E_u^{str}| \right]= E\left[ E\left[|E_u^{str}|-|E_u^{end}| \mid |E_u^{str}| \wedge X=x\right]\right] \geq E\left[\frac{X}{em}\right]\text{.}$$
	We apply linearity of expectation, and use $E[X]\geq 2000em$ to conclude that 
	$$E\left[|E_u^{str}|-|E_u^{end}|\mid |E_u^{str}|\right]\geq E\left[ \frac{X}{em}\right] = \frac{E[X]}{em} \geq  c\text{.}$$
\end{proof}

Using these lemmata, we now prove Theorem~\ref{thm:dyn_arbitrary} and Theorem~\ref{thm:dyn_2approx}, in which our main results are stated for situations where the solution before the first dynamic change is arbitrary, or a maximal matching, respectively.

\begin{theorem}
	Consider the dynamic vertex cover problem where an arbitrary edge is dynamically added to or deleted from the graph with probability $P_D\leq\frac{1}{2000em}$ at each step. Starting with an arbitrary solution $s$, \EBEAE finds a $2$\nobreakdash-approximate solution for this problem in expected time $O(m\log{m})$.
	\label{thm:dyn_arbitrary}
\end{theorem}
\begin{figure}[t]
	\centering
	\includegraphics[width=1\textwidth]{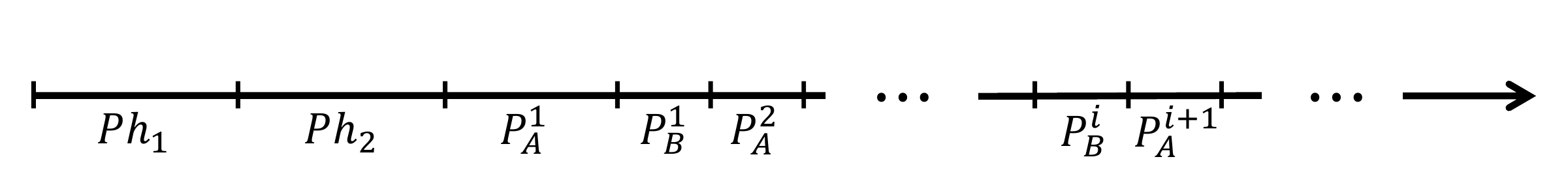}
	\caption{Different phases used in the proofs of Theorems \ref{thm:dyn_arbitrary}, where $Ph_3$ starts with $P_A^1$. Phases $Ph_1$, $Ph_2$ and $P_A^1$ take expected time $O(m\log{m})$, $O(m\log{m})$ and $O(m)$, respectively. Moreover, for all $i\geq 1$ the expected number of steps in sub-phases $P_B^i$ + $P_A^{i+1}$ is in $O(m)$.}
	\label{fig:convex}
\end{figure}
\begin{proof}
	We split the analysis into three phases, $Ph_1$, $Ph_2$ and $Ph_3$, where the last one consists of several sub-phases (Figure~\ref{fig:convex}). The algorithm finds a matching in $Ph_1$, reduces $|C_1|$ to 0 for the first time in $Ph_2$, and finds a maximal matching in $Ph_3$. Jansen et~al.~\cite{DBLP:conf/foga/JansenOZ13} have proved in Theorem~11 of their paper that \EBEAE with their specified fitness function finds a matching in expected time $O(m\log{m})$ for the static version of the problem. In their proof they show that deselecting an edge that shares a node with another selected edge improves the fitness and is always accepted until the selected edges do not share nodes. They show that this problem is easier than OneMax; therefore, the algorithm finds a matching in expected time $O(m\log m)$. The changes that happen on the graph in the dynamic version of the problem (adding an edge or removing an edge) do not increase the number of shared nodes between the selected edges. Therefore, their proof holds for the dynamic vertex cover problem as well, \ie $Ph_1$ needs expected time $O(m\log m)$.
	Moreover, due to the definition of the fitness function, a solution that is not a matching is not accepted after finding a matching; therefore, the algorithm never switches back to $Ph_1$.
	
	Now we analyse $Ph_2$. At the beginning of this phase, the set $C_1$ contains $n'=O(n)$ nodes. Using Lemma~\ref{lem:setting2C1Empty}, we find the expected time to reach $|C_1|=0$ as
	$$\frac{4000em(1+\ln (n'))}{1996} =O(m \log n) \text{.}$$
	
	Now we analyse $Ph_3$. We split this phase into several smaller phases, namely $P^i_A$ and $P^i_B$, for $i \geq 0$. At the beginning of $Ph_3$, the set $C_1$ contains no nodes. We define phases $P^i_A$, such that $|C_1| =0$ holds. When a dynamic change happens, which can increase $|C_1|$, we switch to $P^{i}_B$. Moreover, as soon as the algorithm finds a solution with $|C_1|=0$, we are switched to $P^{i+1}_A$.  Depending on the size of phases $P_A^i$,  $i\geq 0$, we split them to one or more sub-phases, and filter the steps by considering the last step of  these sub-phases. We analyse the expected change that happens on $|E_u|$ in the filtered steps, %In other words, we analyse the changes on $|E_u|$ during $P^i_B$ and $P^{i+1}_A$ and find the expected change on $|E_u|$, from the last step of $P^i_A$ until the last step of $P^{i+1}_A$. By applying a filtration such that only these steps are considered, 
	and use multiplicative drift analysis to find the expected number of required sub-phases until we reach a solution in which $|E_u|$ is a constant. Then we consider the last steps of $P^i_A$ and use additive drift analysis to find the expected time until reaching $|E_u|=0$. Note that, by definition, at filtered steps we also have $|C_1|=0$, because these steps are in $P^i_A$,  $i \geq 0$.
	
	First we find a  lower bound on the expected number of steps of phases $P_A^{i+1}$, $i\geq 0$. We denote the number of steps of a phase $P$ by $|P|$. The expected number of steps of $P_A^i$ and $P_B^i$ are thus denoted by $E[|P_A^i|]$ and $E[|P_B^i|]$, respectively. 
	We find a lower bound on $|P^i_B|+|P^{i+1}_A|$ and an upper bound on $|P_B^i|$ and conclude a lower bound on $|P_A^{i+1}|$. Each phase $P^i_B$, $i\geq 0$, starts with a dynamic change; therefore, from the beginning of phase $P^i_B$, until the beginning of phase $P^{i+1}_B$ we have at least one dynamic change. Note that more than one dynamic change may happen in phase $P^i_B$, but what we need to note here is that the phase $P^{i+1}_B$ cannot start unless we have at least one dynamic change. Pessimistically we assume that the first dynamic change starts $P_B^{i+1}$. Since the probability of a dynamic change is  $P_D\leq \frac{1}{2000em}$, the probability of not starting $P^{i+1}_B$ after $\frac{0.99}{2P_D}$ steps from the start of $P^i_B$ is  
	\begin{eqnarray}
	\label{eq:minProbability}
	(1-{P_D})^\frac{0.99}{2P_D} = \left((1-{P_D})^\frac{0.99}{P_D}\right)^{1/2}\geq  \left(\frac{1}{e}\right)^{1/2}>0.6\text{,}
	\end{eqnarray}
	which implies that with probability at least $0.6$, we have $|P^i_B|+|P^{i+1}_A|\geq \frac{0.99}{2P_D}$. Furthermore, from Lemma~\ref{lem:setting2C1Empty} we know that with probability at least $\frac{5}{6}$, we have $|P^i_B|\leq \frac{4000em(\ln (6/5)+\ln (2))}{1996}$. Altogether, with probability at least $0.6 (5/6)= 0.5$, the number of steps of phase $P^{i+1}_A$ is at least
	\begin{eqnarray*}
		|P^{i+1}_A|\geq \frac{0.99}{2P_D}- \frac{4000em(\ln (6/5)+\ln (2))}{1996}.
	\end{eqnarray*}	
	Since $P_D\leq \frac{1}{2000em}$, we can simplify this inequality to show how many phases of size $em$ fits in the phase $P^{i+1}_A$. Assuming $D=\frac{1}{2000emP_D} $ we find
	\begin{eqnarray*}
		|P^{i+1}_A| &\geq & \frac{0.99(2000emD)}{2}- \frac{4000em(\ln (6/5)+\ln (2))}{1996}\\
		&\geq & em\left(\frac{0.99(2000D)}{2}- \frac{4000(\ln (6/5)+\ln (2))}{1996}\right) \\
		&\geq & em( 990D - 2)\text{.}\\
	\end{eqnarray*}		
	Moreover, since $P_D\leq \frac{1}{2000em}$, we have $D\geq 1$, which gives us
	\begin{eqnarray}
	\label{eq:sizeOfPhaseA}
	|P^{i+1}_A| &\geq & 988Dem
	\end{eqnarray}		
	with probability at least $1/2$. 
	Now we split  phase $P_A^{i+1}$ based on its size:
	\begin{itemize}
		\item If $|P^{i+1}_A| < em$ then we only define one sub-phase, $P^{i+1}_{A_0}$, that is the same size as $P^{i+1}_A$. 
		\item If $|P^{i+1}_A| \geq em$ then for some $l\in \mathbb{N}$ and $0\leq \epsilon<1$ we have $|P^{i+1}_A|= (l+\epsilon) em$. We split $P^{i+1}_A$ to a sub-phase of length $(1+\epsilon)em$, namely $P^{i+1}_{A_0}$, and $l-1$ sub-phases of length $em$, namely $P^{i+1}_{A_k}$, $1\leq k\leq l-1$. 
	\end{itemize}
	Note that from Equation~\ref{eq:sizeOfPhaseA} we can conclude that with probability at least $1/2$ we have $|P_{A_0}^{i+1}|\geq em$ for $i\geq 0$.
	
	We filter the steps and only consider the last step of all of sub-phases $P_{A_k}^{i+1}$, $i\geq 0$ and $k\geq 0$, and denote them by $t^{i+1}_k$. We also denote the value of $|E_u|$ at step $t^{i+1}_k$ by $|E_u^{t^{i+1}_k}|$. Now we find the expected change on $|E_u|$ at the filtered steps. % and use drift analysis to find the expected number of filtered steps that are required until we reach $|E_u|=0$.
	
	We know that the length of sub-phases $P^{i+1}_{A_k}$, $i\geq 1$ and $1\leq k\leq l-1$ is $em$. Therefore, for these sub-phases the condition of Lemma~\ref{lem:driftOnEuDuringSmallPhases} holds and the  drift on $|E_u|$ in filtered steps where $0\leq k\leq l-2$, denoted by $E\left[\Delta^f_{E_u^{t^i_k}}\right]$, is 
	\begin{eqnarray}
	\label{eq:DriftOnSomeOfFilteredSteps}
	E\left[\Delta^f_{E_u^{t^i_k}}\right]\geq E\left[|E_u^{t^{i}_k}|- |E_u^{t^{i}_{k+1}}|\mid |E_u^{t^{i}_k}|\right] \geq (1-1/e)|E_u^{t^i_k}|\text{.}
	\end{eqnarray}
	
	Moreover, we know that with probability at least $1/2$ the condition of Lemma \ref{lem:driftOnEuDuringSmallPhases} holds for sub-phases $P^{i+1}_{A_0}$. Furthermore, we know that the size of $E_u$ is never increased in $P^{i+1}_A$, because we have $|C_1|=0$ and  assumed that no dynamic changes happen  in this phase.  Therefore, we can use the law of total expectation and  find that the expected change on $|E_u|$ during sub-phases of $P^{i+1}_{A_0}$ is at least $ \frac{1}{2}(1-1/e)|E_u^{s}|$,
	where $|E_u^{s}|$ is the value of $|E_u|$ at the last step before the start of $P_{A_0}^{i+1}$.   
	Between the sub-phase $P_{A_{l-1}}^{i}$ and the sub-phase $P_{A_0}^{i+1}$, there exists a phase $P_B^{i}$, in which $|E_u|$ can increase. 
	By Lemma~\ref{lem:IncreaseOnEu}, during phase $P^{i}_B$, the expected increase that happens on $|E_u|$ is upper bounded by $600+\frac{2130}{1996}<602$.
	
	Therefore, together with the expected decrease that happens on $|E_u|$ in $P_{A_0}^{i+1}$, the total drift on $|E_u|$ in the last filtered step of phase $P_A^i$, $i\geq 1$, is 
	$$E\left[\Delta^f_{E_u^{t^i_{l-1}}}\right]=E\left[|E_u^{t^{i}_{l-1}}|- |E_u^{t^{i+1}_{0}}|\mid |E_u^{t^{i}_{l-1}}|\right] \geq  -602 + \frac{1-1/e}{2} \left(|E_u^{t^{i}_{l-1}}|-602\right).$$
	where $l$ is the number of sub-phases of $P_A^i$ and $|E_u^{t^{i}_{l-1}}|-602$ is a lower bound on the value of $|E_u|$ at the beginning of phase $P_{A_0}^{i+1}$.
	
	For $ |E_u^{t^{i}_l}|\geq 8000$, we have $602\leq |E_u^{t^{i}_l}|/13$ and we can have the drift on $|E_u|$ as:
	$$E\left[\Delta^f_{E_u^{t^i_{l-1}}}\right]\geq -\frac{ |E_u^{t^{i}_{l-1}}|}{13} +\frac{ (e-1)}{2e}\left(\frac{12|E_u^{t^{i}_{l-1}}|}{13}\right)$$
	$$\geq \left(-\frac{1}{13} +\frac{ 12(e-1)}{26e}\right)|E_u^{t^{i}_{l-1}}|\geq \frac{ |E_u^{t^{i}_{l-1}}|}{5}\text{.}$$
	
	Together with the drift that we found in Equation~\ref{eq:DriftOnSomeOfFilteredSteps} for other filtered steps, we can conclude that for all filtered steps $t$ we have a drift of at least $\frac{ |E_u^{t}|}{5}$.
	
	Now consider the potential function 
	\begin{equation*}
	g(s)=
	\begin{cases}
	|E_u| & \text{if }  |E_u|\geq 8000\\
	0 & \text{Otherwise}
	\end{cases}\text{.}       
	\end{equation*}
	
	The drift on the value of this function on filtered steps, defined as $E[\Delta_g^f]= E[g(s^{t_k})- g(s^{t_{k+1}})\mid g(s^{t_k})]$ which is  greater than or equal to $ E[\Delta^f_{E_u}]\geq \frac{ |E_u^{t_k}|}{5}$ for $g(s^{t_k})>8000 $ by definition of $g(s)$. This implies that for $g(s^{t_k})>0$
	$$E[\Delta_g^f] \geq  \frac{g(s^{t_k})}{5}\text{.}$$
	Since the minimum value of $g(s)$ before reaching 0 is at least 8000, and the initial value of $g(s)$ is at most $m$, by multiplicative drift analysis we can conclude that in expectation the algorithm needs at most $5(1+\log (m/8000))$ filtered steps to reach $g(s)=0$.
	
	From Lemma~\ref{lem:setting2C1Empty} we know that the expected number of steps of a phase $P_B^i$,  $i\geq 0$, is upper bounded by $\frac{4000em(1+\ln(c))}{1996}$.  Moreover, we know that between each two filtered steps there exists at most one sub-phase of at most $2em$ steps and one phase $P_B^i$, which implies that the expected number of steps between each two filtered steps is upper bounded by $2em+\frac{4000em(1+\ln(c))}{1996}$. Let $X$ and $E(T)$ be the number of filtered steps and the expected total number of steps until the algorithm reaches $g(s)=0$. By the law of total expectation, we have
	$E[T]=E\left[E[T\mid X]\right]\leq E[(2em+\frac{4000em(1+\ln(c))}{1996})\cdot X]$. By linearity of expectation we have $E[T]\leq (2em+\frac{4000em(1+\ln(c))}{1996})\cdot E[X]$. Since $E[X]\leq 5(1+\log (m/8000))$, we have $E[T]\leq (2em+\frac{4000em(1+\ln(c))}{1996})\cdot (5(1+\log (m/8000)))= O(m \log m)$.
	
	% has a length less than $2em$, in an expected number of $O(m \log m)$
	%phase $P_A^i$ consists of $996D$ filtered steps, in an expected number of $O(\frac{\log m}{D})$ phases of $P_A^i$ 
	%we reach a solution with $|E_u|\leq 8000$. %According to Equation~\ref{eq:sizeOfPhaseATogetherWithB}, each of these phases have an expected length of $O(\frac{1}{P_D}+m)= O(mD)$, which means that we reach this solution in expected time $O(m \log m)$. 
	
	So far we have proved that in expected time $O(m \log m)$ we reach $|E_u|\leq 8000$. Now we apply a different filtration and only consider the last steps of phases $P_A^i$, and find the expected number of filtered steps until reaching $|E_u|=0$.  
	Each phase $P^i_B$, $i\geq 0$, starts with a dynamic change; therefore, from the beginning of phase $P^i_B$, until the beginning of phase $P^{i+1}_B$ we have at least one dynamic change, which implies that $E[|P^i_B|+|P^{i+1}_A|]= E[|P^i_B|]+E[|P^{i+1}_A|]\geq \frac{1}{P_D}$.  
	Furthermore, due to Lemma~\ref{lem:setting2C1Empty}, $E[|P^i_B|]\leq \frac{4000em(1+\ln (2))}{1996}$.  Therefore we find the following lower bound for the expected number of steps of $P^{i+1}_A$:
	\begin{eqnarray}
	\label{eq:expectedSizeOfPAi}
	E[|P^{i+1}_A|]  \geq \frac{1}{P_D} -E[|P^i_B|] \geq \frac{1}{P_D} - \frac{4000em(1+\ln (2))}{1996}\geq 1996emD\text{,}
	\end{eqnarray}	
	where $D=\frac{1}{2000emP_D} \geq 1$.
	
	Using Equation~\ref{eq:expectedSizeOfPAi} and Lemma~\ref{lem:driftConstantOnEuDuringSmallPhases} we find that the value of $|E_u|$ is decreased in $P_A^i$  by an expected value of at least $1996D$.
	Moreover, 
	by Lemma~\ref{lem:IncreaseOnEu}, we know that  during phase $P^{i}_B$, the expected increase that happens on $|E_u|$ is upper bounded by $602$.
	Therefore,  the  drift on $|E_u|$ in the last step of phase $P_A^i$, denoted by $E[\Delta^{f_A}_{E_u}]$, is 
	\begin{eqnarray}
	E[\Delta^{f_A}_{E_u}]=E\left[|E_u^{t_k}|- |E_u^{t_{k+1}}|\mid |E_u^{t_k}|\right] \geq 1996D -602\geq 1394D.
	\label{eq:driftOnEu}
	\end{eqnarray}
	
	Starting with $|E_u|\leq 8000$, using additive drift analysis, we find the expected number of phases $P_A^i$ until reaching $|E_u|=0$ to be $\frac{8000}{1394D}$, which implies that from the point we reach $|E_u|\leq 8000$, until $|E_u|=0$ we need $O(1/D)$ phases of $P_B^i+P_A^{i+1}$.
	
	A phase $P_A^{i}$, $i\geq 0$ terminates when a dynamic change happens, which implies that $E[|P_A^{i}|]\leq \frac{1}{P_D}$.
	Furthermore, due to Lemma~\ref{lem:setting2C1Empty}, $E[|P^i_B|]\leq \frac{4000em(1+\ln (2))}{1996}$.  Therefore  the expected number of steps of $P^i_B$ and $P^{i+1}_A$ together is upper bounded by
	\begin{eqnarray*}
		\label{eq:sizeOfPhaseATogetherWithB}
		E[|P^i_B|+|P^{i+1}_A|]=E[|P^i_B|] +E[|P^{i+1}_A|]  \leq \frac{1}{P_D} +\frac{4000em(1+\ln (2))}{1996}= O(mD)\text{.}
	\end{eqnarray*}	
	
	Let $X$ denote the number of phases $P_B^i+P_A^{i+1}$ that the algorithm needs until reaching $|E_u|=0$. Also, let $E[T]$ denote the expected number of steps until reaching this solution. Since the expected length of $P_B^i+P_A^{i+1}$ is $O(mD)$, using the law of total expectation we find 
	$$E[T]=E\left[E[T\mid X]\right]\leq E\left[X\cdot cmD\right]=cmD\cdot E\left[X\right]\text{,}$$ 
	where $c$ is a constant. Moreover, we have found that the expected number of phases until reaching $|E_u|=0$ is $\frac{8000}{1394D}$,  \ie $E[X]\leq \frac{8000}{1394D}$, which implies
	$$E[T]\leq cmD\cdot \frac{8000}{1394D}= O(m)\text{.}$$ 
	
	Since we are only considering the filtered steps, we also know that $|C_1|=0$ holds in that step; therefore, the solution is a maximal matching and therefore a 2-approximation.
	Altogether, the three analysed phases $Ph_1$, $Ph_2$ and $Ph_3$ need expected time $O(m\log m)$.
\end{proof}

\begin{theorem}
	Consider the dynamic vertex cover problem where an arbitrary edge is dynamically added to or deleted from the graph with probability $P_D\leq\frac{1}{2000em}$ at each step. Starting with a $2$\nobreakdash-approximate solution $s$, which is a maximal matching, after a dynamic change, \EBEAE finds a $2$\nobreakdash-approximate solution in expected time $O(m)$.
	\label{thm:dyn_2approx}
\end{theorem}

\begin{proof}
	The proof of this theorem is similar to the proof of Theorem~\ref{thm:dyn_arbitrary}, except that the first couple of phases are no longer required, because dynamic changes do not damage the matching property of the given solution, and the initial values of $| C_1|$ and $| E_u|$ are constants.	
	
	Here, we start with $|C_1|=0$ and $|E_u|=0$, then a dynamic change happens, which can increase  these sizes by at most 2 and 10, respectively. Therefore, the starting situation in this theorem is the same as the last step of phase $P^1_A$ of the proof of Theorem~\ref{thm:dyn_arbitrary}. Moreover, here the initial value of $|E_u|$ is at most 10; therefore, the proof is similar to the proof of Theorem~\ref{thm:dyn_arbitrary} from the point that we reach a solution where $|E_u|$ is a constant. As proved in that theorem by additive drift analysis, in expected time $O(m)$ the algorithm reaches a solution with $|C_1|=0$ and $|E_u|=0$, which is also a maximal matching and induces a 2-approximate solution.
\end{proof}

By Theorems \ref{thm:dyn_arbitrary} and \ref{thm:dyn_2approx}, we have proved that \EBEAE finds the first $2$\nobreakdash-approximate solution for the dynamic vertex cover problem with the probabilistic dynamic changes in expected time $O(m\log{m})$, while the 2-approximation quality can be restored in $O(m)$ after a dynamic change happens on the graph.

\section{Analysis of the Weighted Vertex Cover Problem}
\label{sec:weightedVCP}

%\remark{short description that says first we analyse the first dynamic setting, then the second one}
In this section, we analyse the dynamic version of the weighted vertex cover problem and the behaviour of \oneplusone and \RLS (Algorithms \ref{alg:WeightedEA} and \ref{alg:WeightedRLS}, respectively) in optimising this problem. Similar to Section \ref{sec:classicalVCP}, we begin the analysis by considering the one-time dynamic setting in which the algorithm starts with a 2\nobreakdash-approximate solution and one edge is dynamically added to or removed from the graph. We compute the expected time that is required by the algorithms to find a 2\nobreakdash-approximate solution for the new situation. In this setting we assume that the problem is not subject to more changes before the new 2\nobreakdash-approximate solution is found. This assumption is relaxed in Section~\ref{sec:WVCDyn2}, where multiple dynamic changes can be handled. In the probabilistic dynamic setting that is studied there for the weighted vertex cover problem, one edge is added to or deleted from the graph with probability $P_D$ in each step of the algorithms. 

In both dynamic settings, we assume that an edge that is dynamically added to the graph has an initial weight of 0 in the current solution of the algorithms.
%However, we prove that with the defined fitness function, \oneplusone and \RLS find the optimal solution in $O(w_{\mathrm{max}}\cdot m)$ or $O(OPT\cdot m)$ in situations that start from a 2\nobreakdash-approximate solution that is effected by a dynamic change or from a zero solution that all edges have weight zero initially, respectively. 

\subsection{Re-Optimization Times of \RLS and \oneplusone for the Dynamic Weighted Vertex Cover Problem}
\label{sec:WVCDyn1}
%\remark{do the edits that I sent to you}
In the dynamic setting which is considered in this section, the  algorithms start with a 2\nobreakdash-approximate solution for weighted vertex cover problem (a maximal solution for the dual problem) and then a dynamic change happens \ie either an edge is deleted from the graph or an edge is added to the graph. Our aim is to prove that in expectation or with high probability, \RLS and \oneplusone re\nobreakdash-discover the 2\nobreakdash-approximate solution in $O(w_{\mathrm{max}}\cdot m)$ and $O(\mathrm{OPT}\cdot m+ m^2)$ steps, respectively, where $w_{\mathrm{max}}$ is the maximum weight of a node, and $\mathrm{OPT}$ is the weight of the optimal solution for the weighted vertex cover problem.

We start by presenting two lemmata that prove important properties of the defined fitness function for the weighted vertex cover problem.

\begin{lemma}
	\label{clm:fitnessConst1}
	If the number of violated constraints by the current solution $s$ is $A$, then any solution $s'$ with $B$ constraint violations is rejected by \oneplusone and \RLS if $B>A$.
\end{lemma}

\begin{proof}
	%Recall that 
	%$$f(s)=\vert V_C(s)\vert + (\vert V\vert+1)\cdot \vert \{e\in E\mid e\cap V_C(s)=\emptyset \}\vert$$ 
	Let $U_s$ and $W(s)$ be the number of uncovered edges and the total edge weights for solution $s$, respectively. In addition, let $w_{\mathrm{total}}= \sum_{i=1}^{n}w(v_i)$. 
	Recalling Equation \eqref{equ:wfitfunc}, 
	$$f(s)=W(s)-U_s\cdot (w_{\mathrm{total}}+1)-A\cdot (m+1)\cdot (w_{\mathrm{total}}+1)\text{,}$$ 
	Since $B>A$, we have 
	\begin{eqnarray*}
		f(s')&\leq& W(s^\prime) - U_{s^\prime}\cdot (w_{\mathrm{total}}+1)  - (A+1)\cdot (m+1)\cdot (w_{\mathrm{total}}+1) \\
		&\leq& w_{\mathrm{total}}  -  (A+1)\cdot (m+1)\cdot (w_{\mathrm{total}}+1)\\
		&\leq& w_{\mathrm{total}} -(m+1)\cdot (w_{\mathrm{total}}+1)  -  A\cdot (m+1)\cdot (w_{\mathrm{total}}+1)\\ 
		&\leq&  -m\cdot (w_{\mathrm{total}}+1)  -  A\cdot (m+1)\cdot (w_{\mathrm{total}}+1)\\
		&\leq& f(s)\text{.}
	\end{eqnarray*}
	The second inequality holds because $W(s^\prime)\leq w_{\mathrm{total}}$ and $U_{s^\prime}\geq 0$; whereas the last inequality holds because $W(s)\geq 0$ and $U_s \leq m$.
	Since $f(s)> f(s')$, $s'$ is rejected by \oneplusone and \RLS.
\end{proof}

\begin{lemma}
	\label{clm:fitnessConst2}
	If the current solution $s$ has $U_s$ uncovered edges with no violated constraint, any solution $s'$ where $U_{s^\prime}> U_s$, is rejected by \oneplusone and \RLS.
\end{lemma}

\begin{proof}
	Similar to Lemma \ref{clm:fitnessConst1}, we prove that $f(s)> f(s')$; therefore, $s'$ is rejected by the two algorithms. Since $s$ does not violate the constraints, we have
	$$f(s)=W(s)-U_s\cdot (w_{\mathrm{total}}+1)\text{,}$$ 
	where $W(s)$ and $U_s$ denote the total weight on the edges and the number of uncovered edges of solution $s$, and $w_{\mathrm{total}}= \sum_{i=1}^{n}w(v_i)$. Since $U_{s'}\geq U_s+1$, we have 
	\begin{eqnarray*}
		f(s')&\leq& W(s^\prime) - (U_s +1)\cdot (w_{\mathrm{total}}+1)   \\
		&\leq& w_{\mathrm{total}}  - (w_{\mathrm{total}}+1) - U_s \cdot (w_{\mathrm{total}}+1)\\
		&\leq& - U_s \cdot (w_{\mathrm{total}}+1)\\ 
		&\leq& f(s)\text{.}
	\end{eqnarray*}
	The second inequality holds because $W(s^\prime)\leq w_{\mathrm{total}}$; whereas the last inequality holds because $W(s)\geq 0$.
\end{proof}

In the next lemma, we prove that a dynamic change does not violate any constraints in the current solution.

\begin{lemma}\label{lem:noviolation}
	The dynamic changes of the dynamic weighted vertex cover problem do not cause constraint violations in the dual problem.
\end{lemma}

\begin{proof}
	We first investigate the dynamic addition of an edge. Let $e=\{u,v\}$ be the new edge and $s$ be the current solution. Also, let $w_u$ and $w_v$ denote the total weights of edges adjacent to vertices $u$ and $v$ respectively. Since the initial weight of $e$ is zero, it does not change $w_e$ or $w_v$. Thus, no new violations happens. 
	
	Now assume that $e$ with the weight of $w_e$ is deleted from the graph. In this case, $w_e$ is subtracted from $w_u$ and $w_v$. Since no value is added to them, the weight constraints of the nodes are not violated.
\end{proof}

In this paper we always assume that the algorithms are initialized with either a solution with weight 0 for all edges (which has no constraint violation), or a feasible dual solution which had been maximal before the dynamic change. Moreover, according to Lemmata \ref{clm:fitnessConst1} and \ref{lem:noviolation}, a solution that violates any constraints is not accepted by the algorithms, and the dynamic changes on the graph do not cause violation. Therefore, we assume that from now on, there is no constraint violation on the dual solution that the algorithm is working on.

We first present the analysis for the \RLS. 
Here we give the definitions that help us with our analysis.
Let $W(s)$ denote the sum of edge weights of solution $s$. We define $G$ to be the maximum total amount that should be added to edges of solution $s$ to achieve a maximal solution, \ie:
$$G = \max_{o\in M_s}\left\{W(o)-W(s)\right\}\text{,}$$
where $M_s$ is a set of maximal solutions that can be obtained from $s$ only by increasing the weight of the edges. Moreover, let $m_s\in M_s$ denote the maximal solution that $G$ is obtained from it. Note that it is impossible to add more than $G$ to the edges of solution $s$ without violating a constraint, since it contradicts the definition of $G$. Here the goal is to find a maximal solution for the dual problem to achieve a 2\nobreakdash-approximate weighted vertex cover solution. Therefore, we are seeking a solution with $G=0$.

Let $G_t$ denote the value of $G$ in the step $t$ of the algorithm and $E[\Delta_{t}] = E[G_t-G_{t+1}\mid G_t]$ be the drift on the value of $G$. Also let $k$ denote the number of uncovered edges. 
The following lemma proves a lower bound for the drift on the value of $G$ in a step of the algorithm.

\begin{lemma}\label{lem:wdriftonG}
	At a step $t$ of \RLS, where the current solution has $k$ uncovered edges, the drift on $G$ is $E[\Delta_t]\geq \frac{k}{2m}$.
\end{lemma}
\begin{proof}
	\RLS mutates only one edge at each step and increases it with probability $1/2$. According to Lemmata \ref{clm:fitnessConst1} and \ref{clm:fitnessConst2}, the mutation will not be accepted if it decreases the weight of an edge, since the number of uncovered edges does not decrease by this move and the total weight of the solution decreases. Also, a mutation that increases the weight of a covered edge will be rejected, because a weight constraint on a node becomes violated by this move. Hence, the only accepted mutations are those that increase the value of an uncovered edge, which happen with probability $\frac{k}{2m}$ at each step. Each such move reduces the value of $G$ by one, thus we have $$E[\Delta_t]\geq \frac{k}{2m}\text{.}$$
	
\end{proof}

The following theorem gives the expected time for \RLS to restore the quality of  2\nobreakdash-approximation after a dynamic change has added or deleted an edge to or from the graph of the problem. 
\begin{theorem}
	\label{thm:WVCP-re-opt}
	Starting with a 2\nobreakdash-approximate solution $s$ that is a maximal solution for the dual problem,  \RLS restores the quality of 2-approximation when one edge is dynamically added to or deleted from the graph in expected time $O(w_{\mathrm{max}}\cdot m)$.
\end{theorem}

\begin{proof}
	First we consider the situation in which a new edge $e = \{v_{i},v_{j}\}$ is added to the graph. If any of the adjacent nodes of $e$ are tight, then $e$ is already covered and the new solution is maximal. Therefore assume that neither $v_i$ is tight nor $v_j$. In this case and according to the fact that the solution was maximal before the dynamic change, the weight of edges of $v_i$ and $v_j$ except $e$ can not increase. Therefore, adding 
	$$G = \min\left\{w(v_{i})-\sum_{e_l\in E  \mid \ e_l \cap \{v_{i}\} \neq \emptyset} s_l \ \ \ , \ \ \ w(v_{j})-\sum_{e_l \in E \mid \ e_l \cap \{v_{j}\} \neq \emptyset} s_l\right\}$$ 
	to the weight of the new  edge guarantees achieving a new maximal solution. Since before reaching a maximal solution, the number of uncovered edges is at least one and $G\leq w_{\mathrm{max}}$, Lemma \ref{lem:wdriftonG} and additive drift analysis give the result of $O(w_{\mathrm{max}}\cdot m)$ expected time for re\nobreakdash-optimising the 2\nobreakdash-approximation.
	
	Now consider the edge $e=\{v_i, v_j\}$ with the weight of $w_e\leq w_{\mathrm{max}}$ being deleted from the graph. In this case, by adding at most $w_e$ to the weight of edges that are adjacent to $v_i$ and also to the weight of edges that are adjacent to $v_j$, a maximal solution will be achieved. Moreover, it is not possible to add more than $2w_e$ to the edges of the graph since it contradicts the assumption that $s$ was maximal before the dynamic change. Thus, we have $G\leq 2w_e\leq 2w_{\mathrm{max}}$ and the theorem is proved by using Lemma \ref{lem:wdriftonG} and the additive drift argument. 
\end{proof}

Now we investigate the behaviour of \oneplusone for this problem. We use the definition of $W(s)$ for a solution $s$ which is given earlier in this section, in addition to the weight difference from the maximum dual solution, denoted by $G^*$. If the maximum dual solution is denoted by $o^*$, then the formal definition of $G^*$ would be
$$G^*= W(o^*)-W(s)\text{.}$$ 

Due to the Weak Duality Theorem, $W(o^*)$ is upper bounded by the weight of the optimal solution of the weighted vertex cover problem, denoted by $\mathrm{OPT}$, which implies that for any initial solution $s$, we have $G^*\leq \mathrm{OPT}$.

It can be observed from the definition of the fitness function that $W(s)$ can only be decreased in the steps in which the number of uncovered edges decreases. Let $A$ denote the set of accepted steps during the process of the algorithm and $B$ be the accepted steps in which the number of uncovered edges are decreased. Observe that $B \subseteq A$.

The weight of an edge, if selected for a mutation, is either increased by one, or decreased by one. Let $\Delta^t_{w+}$ and $\Delta^t_{w-}$
denote the number of edges that gain weight and the number of edges that lose weight at an accepted step, $t$. Note that $\Delta^t_{w+}+\Delta^t_{w-}$
is equal to the change that happens on $G^*$ at step $t$. Lemmata~\ref{lem:MinimumChangeOnG*} and~\ref{lem:TotalDecreaseW} show that for all steps in $A\setminus B$ we have $\Delta^t_{w+}+\Delta^t_{w-} \geq 1$ and with high probability the sum of these changes during all steps of $B$ is lower bounded by $-4em$. Moreover, Lemma~\ref{lem:MinimumAcceptedStepsWholePhase} gives a lower bound on $|A\setminus B|$ during a phase of $O(\mathrm{OPT}\cdot m+ m^2)$ steps.
Then in Theorem~\ref{thm:WVCP-re-opt-EA},  we show that after a dynamic change on the graph, with high probability \oneplusone re\nobreakdash-discovers the 2\nobreakdash-approximate solution in  the studied phase.

% Also let $G^*_t$ be the value of $G^*$ at step $t$. We denote the drift on $G^*$ by $E(\Delta^t)=E(G^*_t- G^*_{t+1}\mid G^*_t)= E(\Delta_{w+}^t)-E(\Delta_{w-}^t)$, where $E(\Delta_{w+}^t)$ and  $E(\Delta_{w-}^t)$ are the  positive drift and the  negative drift on $G^*$. In the following lemmata, we show some properties on the expected positive and negative drift of $G^*$, that help us with the proof of

\begin{lemma}
	\label{lem:MinimumChangeOnG*}
	For all steps in $A\setminus B$ we have $\Delta^t_{w+}+\Delta^t_{w-} \geq 1$. 
\end{lemma}
\begin{proof}
	Since the algorithm accepts a move if there is a strict improvement, and at steps $A\setminus B$ the number of uncovered edges does not change, the total weight on the edges increases. \ie 
	\begin{eqnarray*}
		\label{eq:changesOnG*}
		\Delta^t_{w+} - \Delta^t_{w-}\geq 1 \ \ \ \ \forall t \in A\setminus B   \text{.}
	\end{eqnarray*}
\end{proof}

\begin{lemma}
	\label{lem:MinimumAcceptedStepsWholePhase}
	In a phase of $P=2e\mathrm{OPT}\cdot m +10e^2m^2$ steps, if we do not reach a maximal solution, with probability $1-e^{-\Omega(m)}$ the total number of accepted steps is at least $\mathrm{OPT} +5em$.
\end{lemma}
\begin{proof}
	While we have not reached a maximal solution, there exists at least one edge that can be increased in weight without violating a constraint. Therefore, the probability of an accepted move is at least $\frac{1}{em}$ at each step. We define $X_t$ at a step $t$, such that $X_t=1$ if the move of step $t$ is accepted and $X_t=0$ if it is rejected. Using this definition, the total number of accepted steps, denoted by $X$, is 
	$$X=\sum_{t=1 }^{|P|} X_t\text{.}$$ 
	Using the lower bound on the probability of an accepted move at each step, we find the expected value of $X$ as $E[X]\geq 2\mathrm{OPT} +10em$. By Chernoff bounds, with parameter $\delta=1/2$, the probability of $X<\frac{E[X]}{2}$ is
	$$Pr\left(X<\mathrm{OPT} +5em\right) \leq e^{-E[X]\delta^2/2}= e^{-\Omega(m)}\text{,}$$
	which completes the proof.
\end{proof}

\begin{lemma}
	\label{lem:TotalDecreaseW}
	With  probability $1-e^{-\Omega(m)}$ the total decrement on $G^*$ during all steps of $B$ is upper bounded by $ 4em$.
\end{lemma}
\begin{proof}
	Since at least one edge is covered at each step of $B$, at least one node needs to become tight, which implies that  at least one increasing mutation happens on an edge. Therefore, $\forall t\in B, \Delta^t_{w+} \geq 1$, which implies that $\sum_{t\in B} \Delta^t_{w+} \geq |B|$. 
	
	Now we need to prove that with probability $1-e^{-\Omega(m)}$ we have $\sum_{t\in B} \Delta^t_{w-} \leq 4em$, which implies $\sum_{t\in B} \Delta^t_{w+}-\Delta^t_{w-} \geq |B|-4e m \geq -4em$. In other words we prove that the weight of the solution is not decreased more that $4em$ during steps of $B$. Then using the definition of $G^*$, the statement of the theorem is proved. 
	
	In order to prove that with probability $1-e^{-\Omega(m)}$ we have $\sum_{t\in B} \Delta^t_{w-} \leq 4em$, we use an argument similar to the proof of Lemma~\ref{lem:driftOnC1}. We first need to find an upper bound on the probability of mutating each of the other edges, at the same step that an improving move happens.
	
	Let $I$ be the set of moves in which at least one increasing mutation has happened that makes a node tight; hence, decrease the number of uncovered edges. Since the lemma is about the steps in which the number of uncovered edges decrease, we filter the steps based on this, and from this point of the proof, we only consider the steps in which a move from $I$ has happened. 
	
	Together with the increasing mutation, each other edge can also mutate with probability $1/m$. Therefore, at the same step of the increasing mutation(s), some edges may lose weight, but the whole move may be accepted. With some other mutations at the same step, the constraint of some nodes may be violated and the whole move will be rejected. 
	
	Let $Acc$ denote the event that the whole move in a step is accepted.
	Assuming that one move of $I$ has happened, we want to find a lower bound on the probability that this move is accepted. We know that when the improving mutations happen, if no extra mutations happens,  the move will be accepted. Therefore, with probability at least $(1-1/m)^{m-1}\geq 1/e$, a move from $I$ is accepted.  We denote this probability by 
	\begin{eqnarray}
	\label{eq:probAcceptedWeightedCase}
	P_{Acc}\geq \frac{1}{e}\text{.}
	\end{eqnarray}

	Also let $P(\text{bit}\mid Acc)$ denote the probability of mutating an edge $e$, under the condition that event $Acc$ has occurred. By the definition of conditional probability, we know that 
	$$P(\text{bit}\mid Acc)= \frac{P(\text{bit} \cap P_{Acc})}{P_{Acc}}\leq \frac{P(\text{bit})}{P_{Acc}}\text{.}$$
	Using Equation~\ref{eq:probAcceptedWeightedCase}, we get:
	$$P(\text{bit}\mid Acc)\leq eP(\text{bit} )\text{,}$$
	where $P(\text{bit} )$ is the unconditional probability of flipping $e$, which is $\frac{1}{m}$. This implies that
	\begin{eqnarray}
	%\label{eq:Pbit}
	P(\text{bit}\mid Acc)\leq \frac{e}{m}\text{.}
	\end{eqnarray}
	
	Since there are at most $m-1$ edges other than the the edge with improving move, that can mutate at every step of $B$, the expected number of edges that can lose weight during steps of $B$, is at most $\frac{e(m-1)|B|}{m}\leq e|B|$. In other words 
	$$E\left[\sum_{t\in B} \Delta^t_{w-}\right] \leq e|B|\text{.}$$ 
	
	Now we find the probability that $\sum_{t\in B} \Delta^t_{w-}\geq 4em$. Since $|B|$ is upper bounded by $m$, we have 
	$$Pr\left(\sum_{t\in B} \Delta^t_{w-} \geq 4em\right)<Pr\left(\sum_{t\in B} \Delta^t_{w-} \geq \left(\frac{3m}{|B|}+1\right)\cdot e|B|\right)\text{.}$$
	
	Now we use the Chernoff bounds with parameters $\mu = e|B|$ and $\delta=\frac{3m}{|B|} $ to find an upper bound on the right hand side of the inequality above:
	$$Pr\left(\sum_{t\in B} \Delta^t_{w-} \geq (\delta+1)\cdot \mu\right)\leq \left(\frac{e^\delta}{(1+\delta)^{(1+\delta)}} \right)^\mu\leq  \left(\frac{e}{1+\delta} \right)^{\delta \mu}\text{.}$$
	
	Since $|B|$ is bounded by $m$, we have $\delta\geq 3$, which gives us:
	$$Pr\left(\sum_{t\in B} \Delta^t_{w-} \geq (\delta+1)\cdot \mu\right)\leq  \left(\frac{e}{4} \right)^{\frac{3m}{|B|} (e|B|)}= (e/4)^{3em}\text{.}$$
	
	Therefore, with probability at least $1-(e/4)^{3em}= 1-e^{-\Omega(m)}$ we have $\sum_{t\in B} \Delta^t_{w-} \leq 4em$.
\end{proof}

\begin{theorem}
	\label{thm:WVCP-re-opt-EA}
	Starting with a 2-approximate solution that is a maximal solution for the dual problem, after a dynamic addition or deletion of an edge, in a phase of $P=2e\mathrm{OPT}\cdot m +10e^2m^2$ steps, with probability $1-e^{-\Omega(m)}$, the \oneplusone re-rediscovers a 2\nobreakdash-approximation.
\end{theorem}
\begin{proof}
	
	When a new edge is added to the graph, if neither of its nodes are tight in the current solution, then it is not covered by the induced cover set.  Also, when an edge is deleted from the graph, both of its nodes are no longer tight, if the weight of the deleted edge was greater than 0. Therefore, deleting an edge may uncover some other edges. In other words, after both kinds of dynamic changes, the solution does not violate a constraint, but it may not be maximal, and the algorithm needs to find a maximal solution.
	
	Recall that $A$ is the set of all accepted steps during the run of the algorithm and $B\subset A$ is the set of accepted steps in which the number of uncovered edges is decreased. We have $|B|\leq m$, because, according to the fitness function the number of uncovered edges cannot increase during the process. According to Lemma~\ref{lem:MinimumAcceptedStepsWholePhase}, we either reach a maximal solution during phase $P$, in which case the statement of the theorem holds, or with probability $1-e^{-\Omega(m)}$ we have at least $\mathrm{OPT}+5em$ accepted steps. Moreover, we know $|B|\leq m$. Therefore, with probability $1-e^{-\Omega(m)}$ we have $|A\setminus B|\geq \mathrm{OPT+4em}$. 
	
	Also, according to Lemma~\ref{lem:MinimumChangeOnG*}, the minimum increase on the weight of the solution at a step $t\in A\setminus B$ is 1. Therefore, unless we reach a maximal solution, with probability $1-e^{-\Omega(m)}$, the value of $G^*$ is decreased by at least $\mathrm{OPT}+ 4em$ in the phase $P$. Knowing that the initial value of $G^*$ is upper bounded by $OPT$, and the total increment on $G^*$ is upper bounded by $4em$ (Lemma~\ref{lem:TotalDecreaseW}), we can conclude that in the phase $P$, we either reach a maximal solution or a situation where $G^*=0$ which happens when we reach the maximum solution.  In both cases, the obtained solution induces a 2\nobreakdash-approximate vertex cover.
\end{proof}

Note that the proof of Theorem~\ref{thm:WVCP-re-opt-EA} holds with any initial solution that is a feasible solution for the dual problem, and during the process of the algorithm no solution with violations is accepted by the algorithm. This implies that the considered phase, $P$, can be restarted at any point of time; therefore, if we reach the desired solution with probability at least $pr$, we can conclude that the expected time until reaching the desired solution is upper bounded by $\frac{P}{pr}$, which is formalised in the following corollary.

\begin{corollary}
	Starting with a 2-approximate solution that is a maximal solution for the dual problem, after a dynamic addition or deletion of an edge, in expected time $O\left(\mathrm{OPT}\cdot m +m^2\right)$, the \oneplusone re-discovers a solution with the quality of 2\nobreakdash-approximation.
\end{corollary}

\subsection{Weighted Vertex Cover in Probabilistic Dynamic Setting}
\label{sec:WVCDyn2}

In this section we investigate the probabilistic dynamic setting for applying dynamism on the weighted vertex cover problem. Here, with probability $P_D$, a dynamic change happens on the graph at each step. We analyse the behaviour of \RLS (Algorithm~\ref{alg:WeightedRLS}) and  \oneplusone (Algorithm~\ref{alg:WeightedEA}) for $P_D\leq \frac{1}{5w_{\mathrm{max}}em}$, and $P_D\leq \frac{1}{(1+\epsilon)(2e\mathrm{OPT}\cdot m +10e^2m^2)}$, respectively, where $w_{\mathrm{max}}$ is the maximum weight of vertices, $\mathrm{OPT}$ is the weight of the optimal solution for the weighted vertex cover problem, and $\epsilon$ is a small positive constant. 

We first analyse \RLS and start by presenting a lemma that gives us a lower bound on the drift of $G$, $E[\Delta_{t}] = E[G_t-G_{t+1}\mid G_t]$, in this setting. Then we present two theorems that find the expected time that the investigated algorithm need to find 2\nobreakdash-approximate solutions, when the given solution before the dynamic change has a weight of 0 for all the edges (Theorem~\ref{the:wgeneral}) or is a 2\nobreakdash-approximate solution (Theorem~\ref{the:wStartWith2approx}).

\begin{lemma}
	\label{lem:driftWVC2ndSetting}
	Assuming that $P_D\leq \frac{1}{5w_{\mathrm{max}}em}$, at a step $t$ of  \RLS before reaching a maximal solution, the drift on $G$ is $E[\Delta_t]\geq \frac{1}{10em}$.
\end{lemma}
\begin{proof}
	As proven in Lemma \ref{lem:wdriftonG}, the expected improvement that the two investigated algorithms make on $G$ at each step is at least $\frac{k}{2em}$, where $k$ is the number of uncovered edges. In addition to this improvement, here we need to calculate the expected impact of a dynamic change of the graph on the value of $G$, which we denote by $E[\Delta_D]$. In other words, the total drift on $G$ is:
	\begin{align}
	E[\Delta_t]\geq \frac{k}{2em} + E[\Delta_D]\text{.}
	\label{equ:driftWVC2ndSetting}
	\end{align} 
	
	Dynamically adding or deleting an edge increases the value of $G$ by at most $w_{\mathrm{max}}$ and $2w_{\mathrm{max}}$ respectively (See the argument in the proof of Theorem~\ref{thm:WVCP-re-opt}). Hence we have:
	$$E[\Delta_D]\geq -2w_{\mathrm{max}}\cdot P_D=\frac{-2}{5em}\text{.}$$
	Together with Equation~\eqref{equ:driftWVC2ndSetting}, and the fact that $k\geq 1$ when we have not reached a maximal solution, we can conclude that the total drift on $G$ is: $$E[\Delta_t]\geq\frac{5k-4}{10em}\geq \frac{1}{10em}\text{.}$$
\end{proof}

\begin{theorem}\label{the:wgeneral}
	Consider the dynamic weighted vertex cover problem where an arbitrary edge is dynamically added to or deleted from the graph with probability $P_D\leq \frac{1}{5w_{\mathrm{max}}em}$ at each step. Starting with a solution $s= (0, \cdots, 0)$, \RLS finds a 2\nobreakdash-approximate solution for this problem in expected time $O(m\cdot \mathrm{OPT})$.
\end{theorem}
\begin{proof}
	%\remark{I guess we can do something about the violation as well. But we are late!}
	
	Since the value of $G$ for an initial solution is less than or equal to $ \mathrm{OPT}$, using Lemma~\ref{lem:driftWVC2ndSetting} and the additive drift argument we conclude that the expected time to achieve a maximal solution is $O(m\cdot \mathrm{OPT})$.
\end{proof}

\begin{theorem}
	\label{the:wStartWith2approx}
	Consider the dynamic weighted vertex cover problem where an arbitrary edge is dynamically added to or deleted from the graph with probability $P_D\leq \frac{1}{5w_{\mathrm{max}}em}$ at each step. Starting with a 2\nobreakdash-approximate solution $s$ that is maximal for the dual problem, \RLS restores the quality of 2-approximation in expected time $O(w_{\mathrm{max}}\cdot m)\text{.}$
\end{theorem}
\begin{proof}
	From the argument in the proof of Theorem~\ref{thm:WVCP-re-opt}, we know that dynamically adding or deleting an edge increases the value of $G$ by at most $w_{\mathrm{max}}$ and $2w_{\mathrm{max}}$ respectively. Therefore, for the the initial solution we have $G\leq 2w_{\mathrm{max}}$. Similar to Theorem \ref{the:wgeneral}, using Lemma~\ref{lem:driftWVC2ndSetting} and additive drift, the expected time for \RLS to achieve a maximal solution for the dual problem is $O(w_{\mathrm{max}}\cdot m)$\text{.}
\end{proof}

Now we investigate the behaviour of \oneplusone in the probabilistic dynamic setting. In addition to the definition of $W(s)$ for a solution $s$, we use the definitions of $G^*$, $A$ and $B$ that we had in Section~\ref{sec:WVCDyn1} for analysis of \oneplusone. Since the maximal solution that the algorithm moves towards can change in the middle of the process, similar to the proof of  Theorem~\ref{thm:WVCP-re-opt-EA}, we consider the distance to the maximum dual solution. Therefore, the initial distance, after a dynamic change can be in $\Omega(\mathrm{OPT})$, even when we start with a 2-approximate solution. We here analyse the \oneplusone when the initial solution is a solution that is feasible for the dual problem, which means that the result holds for a solution with weight 0, as well as a solution that had been a maximal dual solution before the dynamic change.

\begin{theorem}
	\label{the:setting2wOneplusone}
	Consider the dynamic weighted vertex cover problem where an arbitrary edge is dynamically added to or deleted from the graph with probability $P_D< \frac{1}{(1+\epsilon)(2e\mathrm{OPT}\cdot m +10e^2m^2)}$ at each step, where $\epsilon>0$ is a small constant. Starting with a solution that is feasible for the dual problem, after a dynamic change \oneplusone  finds a 2-approximate solution in a phase of $P=\alpha (1+\epsilon) (2e\mathrm{OPT}\cdot m +10e^2m^2)$ steps, $\alpha>0$ a positive number, with probability $\left(1-(\frac{e-1}{e})^\alpha\right)\left(1-e^{-\Omega(m)}\right)$.
\end{theorem}
\begin{proof}
	From Theorem~\ref{thm:WVCP-re-opt-EA}, we know that starting with a 2-approximate solution, if a dynamic change happens that destroys the quality of 2-approximation, with probability $1-e^{-\Omega(m)}$, the algorithm re-discovers a 2-approximate solution in a phase of $P=2e\mathrm{OPT}\cdot m +10e^2m^2$ steps. The proof of this theorem holds for any initial solution that is a feasible solution for the dual problem. 
	
	Moreover, if the algorithm starts with a solution that is feasible for the dual problem, then it will always be feasible, because the fitness function is defined such that an increase in violation is never accepted by the algorithm, and according to Lemma~\ref{lem:noviolation} dynamic changes do not cause violations either. Therefore, if the initial solution is a feasible solution for the dual problem, then from any point of the process of the algorithm, in a phase of $P=2e\mathrm{OPT}\cdot m +10e^2m^2$ steps with no dynamic changes, with probability $1-e^{-\Omega(m)}$ the algorithm finds a 2-approximate solution.
	
	Since the probability of a dynamic change is $P_D< \frac{1}{(1+\epsilon)P}$ at each step, the probability of not facing any dynamic changes in a phase of $P$ steps is at least
	$$\left(1-P_D\right)^{P}\ge \frac{1}{e}\text{.}$$
	If we have $\alpha$ phases of size $P$, then the probability of facing at least one dynamic change in all of them is at most
	$$\left(1-\frac{1}{e}\right)^\alpha\text{,}$$
	which implies that with probability $1-\left(\frac{e-1}{e}\right)^\alpha$, at least one of these phases does not include a dynamic change, and due to Theorem~\ref{thm:WVCP-re-opt-EA} with probability $1-e^{-\Omega(m)}$, a 2-approximate solution is reached in that phase. 
\end{proof}

%The proof of this theorem is similar to the proof of Theorem~\ref{thm:WVCP-re-opt-EA}, except that at each step, with a probability $P_D$ some dynamic change happens, which increases $G^*$ by a value, bounded above by $\mathrm{OPT}$. In the proof of Theorem~\ref{thm:WVCP-re-opt-EA} we showed that while there exists at least one uncovered edge, the probability of an step to be accepted is at least $\frac{1}{em}$. Therefore, in expectation, there are at most $em$ edges between each two consecutive accepted steps, among which, the expected number of dynamic changes are $emP_D$. As a result, the drift on $G^*$ at an accepted step $t\in A$ would be:
%$$E(\Delta^t)=  E(\Delta^t_{w+}) - E(\Delta^t_{w-}) - emP_D\cdot \mathrm{OPT} \geq \mid A\mid (1-emP_D\cdot \mathrm{OPT}) -em$$
%Since $P_D\leq \frac{1}{2em\cdot \mathrm{OPT}}$ we have
%$$E(\Delta^t)\geq \mid A\mid (1-\frac{1}{2}) -em$$
%With similar argument to what we had in the proof of Theorem~\ref{thm:WVCP-re-opt-EA}, we can find that the expected time until finding a 2-approximate solution  is $O(m\cdot \mathrm{OPT} +m^2)$.
%\end{proof}
The following corollary can be obtained from Theorem~\ref{the:setting2wOneplusone}, in which the expected time to find a 2-approximate solution is found.
\begin{corollary}
	Consider the dynamic weighted vertex cover problem where an arbitrary edge is dynamically added to or deleted from the graph with probability $P_D< \frac{1}{(1+\epsilon)(2e\mathrm{OPT}\cdot m +10e^2m^2)}$ at each step, $\epsilon>0$. Starting with a solution that is feasible for the dual problem,  \oneplusone  finds a 2-approximate solution in expected time $O(\mathrm{OPT}\cdot m +m^2)$.
\end{corollary}

\section{Conclusion}
\label{sec::conclusion}
In this paper, we have investigated the behaviour of two simple randomised search heuristics, namely \RLS and \oneplusone, on the dynamic vertex cover problem. The expected required time for re-optimization of the classical vertex cover problem after a dynamic change had been analysed in~\cite{UsDVCGecco2015} for these two algorithms with an edge-based representation. In this paper we have improved the results of the analysis of \oneplusone. Moreover, we have investigated the probabilistic dynamic setting for applying dynamism to the problem, in which the instance of the problem is subject to a change at each step with probability $P_D$. We have proved that starting from an arbitrary solution, \oneplusone achieves a maximal matching which is a $2$\nobreakdash-approximate solution in expected time $\mathcal{O}(m\log{m})$, when the dynamic changes take place with probability $P_D\leq \frac{1}{2000em}$. Furthermore, with the same setting, we have proved that the  investigated algorithm restores the quality of $2$\nobreakdash-approximation, in expected time of $O(m)$.

After the analysis for the classical vertex cover problem, we have turned our attention to the weighted vertex cover problem, in which weights of at most $w_{\mathrm{max}}$ are assigned to vertices. We have used a similar edge-based representation and investigated both dynamic settings for this problem for \RLS and \oneplusone. In the one-time dynamic setting, assuming that the algorithms are initialized with a 2\nobreakdash-approximate solution before a dynamic change happens on the graph, we have found the expected times of $O(w_{\mathrm{max}}\cdot m)$ and  $O(\mathrm{OPT}\cdot m+ m^2)$ to restore the quality of $2$\nobreakdash-approximation after the change for \RLS and \oneplusone, respectively. In the probabilistic dynamic setting,  for \RLS and \oneplusone, we have assumed that a dynamic change happens at each step with probability $P_D\leq \frac{1}{5w_{\mathrm{max}}em}$ and  $P_D< \frac{1}{(1+\epsilon)(2e\mathrm{OPT}\cdot m +10e^2m^2)}$, respectively. With these assumptions we have proved that starting from a 2\nobreakdash-approximate solution, \RLS re-\nobreakdash optimises the solution in expected times $O(w_{\mathrm{max}}\cdot m)$. We have also proved that in this setting, starting from a solution that assigns a weight of 0 to all edges, \RLS and \oneplusone find a 2\nobreakdash-approximate solution in $O(m\cdot \mathrm{OPT})$ and $O(m\cdot \mathrm{OPT}+ m^2)$, respectively, where $\mathrm{OPT}$ is the weight of the optimal vertex cover.

\section*{Acknowledgment}
This research has been supported by the Australian Research Council (ARC) grants DP140103400 and DP160102401.

%% The Appendices part is started with the command \appendix;
%% appendix sections are then done as normal sections
%% \appendix

%% \section{}
%% \label{}

%% If you have bibdatabase file and want bibtex to generate the
%% bibitems, please use
%%

\end{document}